\documentclass[conference]{IEEEtran}
\usepackage[nolist,withpage]{acronym}
\usepackage{etoolbox}
\makeatletter
\newif\if@in@acrolist
\AtBeginEnvironment{acronym}{\@in@acrolisttrue}
\newrobustcmd{\LU}[2]{\if@in@acrolist#1\else#2\fi}

\newcommand{\ACF}[1]{{\@in@acrolisttrue\acf{#1}}}
\makeatother

\usepackage{amsmath,amsthm,amssymb}
\newcommand\myoverset[2]{\overset{\textstyle #1\mathstrut}{#2}}
\usepackage{mathrsfs,mathtools,mathabx}
\usepackage{bm}
\usepackage{graphicx}
\usepackage{dsfont,enumitem,stackrel}
\usepackage{color}
\usepackage{array}
\usepackage{cases}

\usepackage{algorithm,algpseudocode}
\usepackage{algpseudocode}
\usepackage{pgfplots}
\pgfplotsset{compat=newest}
\usepackage{accents}
\usepackage{subfigure}
\usepackage{float}
\usepackage{cite}
\usepackage{xparse}
\usepackage{xcolor}
\usepackage{tikz}
\usetikzlibrary{spy}
\usetikzlibrary{arrows,decorations.pathreplacing,patterns}
\usepackage[bookmarks=false]{hyperref}
\usepackage{booktabs} 

\definecolor{antiquewhite}{rgb}{0.98, 0.92, 0.84}
\definecolor{antiquefuchsia}{rgb}{0.57, 0.36, 0.51}\definecolor{chestnut}{rgb}{0.8, 0.36, 0.36}
\definecolor{airforceblue}{rgb}{0.36, 0.54, 0.66}
\definecolor{cadmiumorange}{rgb}{0.93, 0.53, 0.18}
\definecolor{bleudefrance}{rgb}{0.19, 0.55, 0.91}
\definecolor{carolinablue}{rgb}{0.6, 0.73, 0.89}
\definecolor{blue(ncs)}{rgb}{0.0, 0.53, 0.74}
\definecolor{dodgerblue}{rgb}{0.12, 0.56, 1.0}
\definecolor{cssgreen}{rgb}{0.0, 0.5, 0.0}
\definecolor{cadmiumgreen}{rgb}{0.0, 0.42, 0.24}
\definecolor{cadmiumorange}{rgb}{0.93, 0.53, 0.18}
\definecolor{amaranth}{rgb}{0.9, 0.17, 0.31}
\definecolor{bluegray}{rgb}{0.4, 0.6, 0.8}
\definecolor{cadmiumgreen}{rgb}{0.0, 0.42, 0.24}

\graphicspath{{./Figures/}} 
\newtheorem{prop}{Proposition}

\newtheorem{lem}{Lemma}	
\newtheorem{cor}{Corollary}
\theoremstyle{definition}

\begin{document}
%\chapter*{Acronyms}
%\renewcommand{\titulonome}{Acronyms}%
%\renewcommand{\prepbynome}{UFC.33 Team}%

%\begin{singlespace}
\begin{acronym}[LTE-Advanced]%\addtolength{\itemsep}{-0.5\baselineskip}
  \acro{2G}{Second Generation}
  \acro{6G}{Sixth Generation}
  \acro{3-DAP}{3-Dimensional Assignment Problem}
  \acro{AA}{Antenna Array}
  \acro{AC}{Admission Control}
  \acro{AD}{Attack-Decay}
  \acro{ADC}{analog-to-digital conversion}
  \acro{ADMM}{alternating direction method of multipliers}
  \acro{ADSL}{Asymmetric Digital Subscriber Line}
  \acro{AHW}{Alternate Hop-and-Wait}
  \acro{AI}{Artificial Intelligence}
  \acro{AirComp}{Over-the-air computation}
  \acro{AM}{amplitude modulation}
  \acro{AMC}{Adaptive Modulation and Coding}
  \acro{AP}{\LU{A}{a}ccess \LU{P}{p}oint}
  \acro{APA}{Adaptive Power Allocation}
  \acro{ARMA}{Autoregressive Moving Average}
  \acro{ARQ}{\LU{A}{a}utomatic \LU{R}{r}epeat \LU{R}{r}equest}
  \acro{ATES}{Adaptive Throughput-based Efficiency-Satisfaction Trade-Off}
  \acro{AWGN}{additive white Gaussian noise}
  \acro{BAA}{\LU{B}{b}roadband \LU{A}{a}nalog \LU{A}{a}ggregation}
  \acro{BB}{Branch and Bound}
  \acro{BCD}{block coordinate descent}
  \acro{BD}{Block Diagonalization}
  \acro{BER}{Bit Error Rate}
  \acro{BF}{Best Fit}
  \acro{BFD}{bidirectional full duplex}
  \acro{BLER}{BLock Error Rate}
  \acro{BPC}{Binary Power Control}
  \acro{BPSK}{Binary Phase-Shift Keying}
  \acro{BRA}{Balanced Random Allocation}
  \acro{BS}{base station}
  \acro{BSUM}{block successive upper-bound minimization}
  \acro{CAP}{Combinatorial Allocation Problem}
  \acro{CAPEX}{Capital Expenditure}
  \acro{CBF}{Coordinated Beamforming}
  \acro{CBR}{Constant Bit Rate}
  \acro{CBS}{Class Based Scheduling}
  \acro{CC}{Congestion Control}
  \acro{CDF}{Cumulative Distribution Function}
  \acro{CDMA}{Code-Division Multiple Access}
  \acro{CE}{\LU{C}{c}hannel \LU{E}{e}stimation}
  \acro{CL}{Closed Loop}
  \acro{CLPC}{Closed Loop Power Control}
  \acro{CML}{centralized machine learning}
  \acro{CNR}{Channel-to-Noise Ratio}
  \acro{CNN}{\LU{C}{c}onvolutional \LU{N}{n}eural \LU{N}{n}etwork}
  \acro{CPA}{Cellular Protection Algorithm}
  \acro{CPICH}{Common Pilot Channel}
  \acro{CoCoA}{\LU{C}{c}ommunication efficient distributed dual \LU{C}{c}oordinate \LU{A}{a}scent}
  \acro{CoMAC}{\LU{C}{c}omputation over \LU{M}{m}ultiple-\LU{A}{a}ccess \LU{C}{c}hannels}
  \acro{CoMP}{Coordinated Multi-Point}
  \acro{CQI}{Channel Quality Indicator}
  \acro{CRM}{Constrained Rate Maximization}
	\acro{CRN}{Cognitive Radio Network}
  \acro{CS}{Coordinated Scheduling}
  \acro{CSI}{\LU{C}{c}hannel \LU{S}{s}tate \LU{I}{i}nformation}
  \acro{CSMA}{\LU{C}{c}arrier \LU{S}{s}ense \LU{M}{m}ultiple \LU{A}{a}ccess}
  \acro{CUE}{Cellular User Equipment}
  \acro{D2D}{device-to-device}
  \acro{DAC}{digital-to-analog converter}
  \acro{DC}{direct current}
  \acro{DCA}{Dynamic Channel Allocation}
  \acro{DE}{Differential Evolution}
  \acro{DFT}{Discrete Fourier Transform}
%  \acro{DIST}{Distance-based Grouping}
  \acro{DIST}{Distance}
  \acro{DL}{downlink}
  \acro{DMA}{Double Moving Average}
  \acro{DML}{Distributed Machine Learning}
  \acro{DMRS}{demodulation reference signal}
  \acro{D2DM}{D2D Mode}
  \acro{DMS}{D2D Mode Selection}
  \acro{DNN}{Deep Neural Network}
  \acro{DPC}{Dirty Paper Coding}
  \acro{DRA}{Dynamic Resource Assignment}
  \acro{DSA}{Dynamic Spectrum Access}
  \acro{DSGD}{\LU{D}{d}istributed \LU{S}{s}tochastic \LU{G}{g}radient \LU{D}{d}escent}
  \acro{DSM}{Delay-based Satisfaction Maximization}
  \acro{ECC}{Electronic Communications Committee}
  \acro{EFLC}{Error Feedback Based Load Control}
  \acro{EI}{Efficiency Indicator}
  \acro{eNB}{Evolved Node B}
  \acro{EPA}{Equal Power Allocation}
  \acro{EPC}{Evolved Packet Core}
  \acro{EPS}{Evolved Packet System}
  \acro{E-UTRAN}{Evolved Universal Terrestrial Radio Access Network}
  \acro{ES}{Exhaustive Search}
  %\acro{FD}{full duplex}
  \acro{FC}{\LU{F}{f}usion \LU{C}{c}enter}
  \acro{FD}{\LU{F}{f}ederated \LU{D}{d}istillation}
  \acro{FDD}{frequency division duplex}
  \acro{FDM}{Frequency Division Multiplexing}
  \acro{FDMA}{\LU{F}{f}requency \LU{D}{d}ivision \LU{M}{m}ultiple \LU{A}{a}ccess}
  \acro{FedAvg}{\LU{F}{f}ederated \LU{A}{a}veraging}
  \acro{FER}{Frame Erasure Rate}
  \acro{FF}{Fast Fading}
  \acro{FL}{Federated Learning}
  \acro{FSB}{Fixed Switched Beamforming}
  \acro{FST}{Fixed SNR Target}
  \acro{FTP}{File Transfer Protocol}
  \acro{GA}{Genetic Algorithm}
  \acro{GBR}{Guaranteed Bit Rate}
  \acro{GD}{gradient descent}
  \acro{GLR}{Gain to Leakage Ratio}
  \acro{GOS}{Generated Orthogonal Sequence}
  \acro{GPL}{GNU General Public License}
  \acro{GRP}{Grouping}
  \acro{HARQ}{Hybrid Automatic Repeat Request}
  \acro{HD}{half-duplex}
  \acro{HMS}{Harmonic Mode Selection}
  \acro{HOL}{Head Of Line}
  \acro{HSDPA}{High-Speed Downlink Packet Access}
  \acro{HSPA}{High Speed Packet Access}
  \acro{HTTP}{HyperText Transfer Protocol}
  \acro{ICMP}{Internet Control Message Protocol}
  \acro{ICI}{Intercell Interference}
  \acro{ID}{Identification}
  \acro{IETF}{Internet Engineering Task Force}
  \acro{ILP}{Integer Linear Program}
  \acro{JRAPAP}{Joint RB Assignment and Power Allocation Problem}
  \acro{UID}{Unique Identification}
  \acro{IID}{\LU{I}{i}ndependent and \LU{I}{i}dentically \LU{D}{d}istributed}
  \acro{IIR}{Infinite Impulse Response}
  \acro{ILP}{Integer Linear Problem}
  \acro{IMT}{International Mobile Telecommunications}
  \acro{INV}{Inverted Norm-based Grouping}
  \acro{IoT}{Internet of Things}
%  \acro{IP}{Internet Protocol}
  \acro{IP}{Integer Programming}
  \acro{IPv6}{Internet Protocol Version 6}
  \acro{IQ}{in-phase quadrature}
  \acro{ISD}{Inter-Site Distance}
  \acro{ISI}{Inter Symbol Interference}
  \acro{ITU}{International Telecommunication Union}
  \acro{JAFM}{joint assignment and fairness maximization}
  \acro{JAFMA}{joint assignment and fairness maximization algorithm}
  \acro{JOAS}{Joint Opportunistic Assignment and Scheduling}
  \acro{JOS}{Joint Opportunistic Scheduling}
  \acro{JP}{Joint Processing}
	\acro{JS}{Jump-Stay}
  \acro{KKT}{Karush-Kuhn-Tucker}
  \acro{L3}{Layer-3}
  \acro{LAC}{Link Admission Control}
  \acro{LA}{Link Adaptation}
  \acro{LC}{Load Control}
  \acro{LDC}{\LU{L}{l}earning-\LU{D}{d}riven \LU{C}{c}ommunication}
  \acro{LOS}{line of sight}
  \acro{LP}{Linear Programming}
  \acro{LTE}{Long Term Evolution}
	\acro{LTE-A}{\ac{LTE}-Advanced}
  \acro{LTE-Advanced}{Long Term Evolution Advanced}
  \acro{M2M}{Machine-to-Machine}
  \acro{MAC}{multiple access channel}
  \acro{MANET}{Mobile Ad hoc Network}
  \acro{MC}{Modular Clock}
  \acro{MCS}{Modulation and Coding Scheme}
  \acro{MDB}{Measured Delay Based}
  \acro{MDI}{Minimum D2D Interference}
  \acro{MF}{Matched Filter}
  \acro{MG}{Maximum Gain}
  \acro{MH}{Multi-Hop}
  \acro{MIMO}{\LU{M}{m}ultiple \LU{I}{i}nput \LU{M}{m}ultiple \LU{O}{o}utput}
  \acro{MINLP}{mixed integer nonlinear programming}
  \acro{MIP}{Mixed Integer Programming}
  \acro{MISO}{multiple input single output}
  \acro{ML}{Machine Learning}
  \acro{MLWDF}{Modified Largest Weighted Delay First}
  \acro{MME}{Mobility Management Entity}
  \acro{MMSE}{minimum mean squared error}
  \acro{MOS}{Mean Opinion Score}
  \acro{MPF}{Multicarrier Proportional Fair}
  \acro{MRA}{Maximum Rate Allocation}
  \acro{MR}{Maximum Rate}
  \acro{MRC}{Maximum Ratio Combining}
  \acro{MRT}{Maximum Ratio Transmission}
  \acro{MRUS}{Maximum Rate with User Satisfaction}
  \acro{MS}{Mode Selection}
  \acro{MSE}{\LU{M}{m}ean \LU{S}{s}quared \LU{E}{e}rror}
  \acro{MSI}{Multi-Stream Interference}
  \acro{MTC}{Machine-Type Communication}
  \acro{MTSI}{Multimedia Telephony Services over IMS}
  \acro{MTSM}{Modified Throughput-based Satisfaction Maximization}
  \acro{MU-MIMO}{Multi-User Multiple Input Multiple Output}
  \acro{MU}{Multi-User}
  \acro{NAS}{Non-Access Stratum}
  \acro{NB}{Node B}
	\acro{NCL}{Neighbor Cell List}
  \acro{NLP}{Nonlinear Programming}
  \acro{NLOS}{non-line of sight}
  \acro{NMSE}{Normalized Mean Square Error}
  \acro{NN}{Neural Network}
  \acro{NOMA}{\LU{N}{n}on-\LU{O}{o}rthogonal \LU{M}{m}ultiple \LU{A}{a}ccess}
  \acro{NORM}{Normalized Projection-based Grouping}
  \acro{NP}{non-polynomial time}
  \acro{NRT}{Non-Real Time}
  \acro{NSPS}{National Security and Public Safety Services}
  \acro{O2I}{Outdoor to Indoor}
  \acro{OAC}{over-the-air computation}
  \acro{OFDMA}{\LU{O}{o}rthogonal \LU{F}{f}requency \LU{D}{d}ivision \LU{M}{m}ultiple \LU{A}{a}ccess}
  \acro{OFDM}{Orthogonal Frequency Division Multiplexing}
  \acro{OFPC}{Open Loop with Fractional Path Loss Compensation}
	\acro{O2I}{Outdoor-to-Indoor}
  \acro{OL}{Open Loop}
  \acro{OLPC}{Open-Loop Power Control}
  \acro{OL-PC}{Open-Loop Power Control}
  \acro{OPEX}{Operational Expenditure}
  \acro{ORB}{Orthogonal Random Beamforming}
  \acro{JO-PF}{Joint Opportunistic Proportional Fair}
  \acro{OSI}{Open Systems Interconnection}
  \acro{PAIR}{D2D Pair Gain-based Grouping}
  \acro{PAPR}{Peak-to-Average Power Ratio}
  \acro{P2P}{Peer-to-Peer}
  \acro{PC}{Power Control}
  \acro{PCI}{Physical Cell ID}
  \acro{PDCCH}{physical downlink control channel}
  \acro{PDD}{penalty dual decomposition}
  \acro{PDF}{Probability Density Function}
  \acro{PER}{Packet Error Rate}
  \acro{PF}{Proportional Fair}
  \acro{P-GW}{Packet Data Network Gateway}
  \acro{PL}{Pathloss}
  \acro{PLL}{phase-locked Loop}
  \acro{PRB}{Physical Resource Block}
  \acro{PROJ}{Projection-based Grouping}
  \acro{ProSe}{Proximity Services}
%  \acro{PS}{Packet Scheduling}
%  \acro{PS}{phase shifter}
  \acro{PS}{\LU{P}{p}arameter \LU{S}{s}erver}
  \acro{PSO}{Particle Swarm Optimization}
  \acro{PUCCH}{physical uplink control channel}
  \acro{PZF}{Projected Zero-Forcing}
  \acro{QAM}{Quadrature Amplitude Modulation}
  \acro{QoS}{quality of service}
  \acro{QPSK}{Quadri-Phase Shift Keying}
  \acro{RAISES}{Reallocation-based Assignment for Improved Spectral Efficiency and Satisfaction}
  \acro{RAN}{Radio Access Network}
  \acro{RA}{Resource Allocation}
  \acro{RAT}{Radio Access Technology}
  \acro{RATE}{Rate-based}
  \acro{RB}{resource block}
  \acro{RBG}{Resource Block Group}
  \acro{REF}{Reference Grouping}
  \acro{RF}{radio frequency}
  \acro{RLC}{Radio Link Control}
  \acro{RM}{Rate Maximization}
  \acro{RNC}{Radio Network Controller}
  \acro{RND}{Random Grouping}
  \acro{RRA}{Radio Resource Allocation}
  \acro{RRM}{\LU{R}{r}adio \LU{R}{r}esource \LU{M}{m}anagement}
  \acro{RSCP}{Received Signal Code Power}
  \acro{RSRP}{reference signal receive power}
  \acro{RSRQ}{Reference Signal Receive Quality}
  \acro{RR}{Round Robin}
  \acro{RRC}{Radio Resource Control}
  \acro{RSSI}{received signal strength indicator}
  \acro{RT}{Real Time}
  \acro{RU}{Resource Unit}
  \acro{RUNE}{RUdimentary Network Emulator}
  \acro{RV}{Random Variable}
  \acro{SAC}{Session Admission Control}
  \acro{SCM}{Spatial Channel Model}
  \acro{SC-FDMA}{Single Carrier - Frequency Division Multiple Access}
  \acro{SD}{Soft Dropping}
  \acro{S-D}{Source-Destination}
  \acro{SDPC}{Soft Dropping Power Control}
  \acro{SDMA}{Space-Division Multiple Access}
  \acro{SDR}{semidefinite relaxation}
  \acro{SDP}{semidefinite programming}
  \acro{SER}{Symbol Error Rate}
  \acro{SES}{Simple Exponential Smoothing}
  \acro{S-GW}{Serving Gateway}
  \acro{SGD}{\LU{S}{s}tochastic \LU{G}{g}radient \LU{D}{d}escent}  
  \acro{SINR}{signal-to-interference-plus-noise ratio}
%   \acro{SI}{Satisfaction Indicator}
  \acro{SI}{self-interference}
  \acro{SIP}{Session Initiation Protocol}
  \acro{SISO}{\LU{S}{s}ingle \LU{I}{i}nput \LU{S}{s}ingle \LU{O}{o}utput}
  \acro{SIMO}{Single Input Multiple Output}
  \acro{SIR}{Signal to Interference Ratio}
  \acro{SLNR}{Signal-to-Leakage-plus-Noise Ratio}
  \acro{SMA}{Simple Moving Average}
  \acro{SNR}{\LU{S}{s}ignal-to-\LU{N}{n}oise \LU{R}{r}atio}
  \acro{SORA}{Satisfaction Oriented Resource Allocation}
  \acro{SORA-NRT}{Satisfaction-Oriented Resource Allocation for Non-Real Time Services}
  \acro{SORA-RT}{Satisfaction-Oriented Resource Allocation for Real Time Services}
  \acro{SPF}{Single-Carrier Proportional Fair}
  \acro{SRA}{Sequential Removal Algorithm}
  \acro{SRS}{sounding reference signal}
  \acro{SU-MIMO}{Single-User Multiple Input Multiple Output}
  \acro{SU}{Single-User}
  \acro{SVD}{Singular Value Decomposition}
  \acro{SVM}{\LU{S}{s}upport \LU{V}{v}ector \LU{M}{m}achine}
  \acro{TCP}{Transmission Control Protocol}
  \acro{TDD}{time division duplex}
  \acro{TDMA}{\LU{T}{t}ime \LU{D}{d}ivision \LU{M}{m}ultiple \LU{A}{a}ccess}
  \acro{TNFD}{three node full duplex}
  \acro{TETRA}{Terrestrial Trunked Radio}
  \acro{TP}{Transmit Power}
  \acro{TPC}{Transmit Power Control}
  \acro{TTI}{transmission time interval}
  \acro{TTR}{Time-To-Rendezvous}
  \acro{TSM}{Throughput-based Satisfaction Maximization}
  \acro{TU}{Typical Urban}
  \acro{UE}{\LU{U}{u}ser \LU{E}{e}quipment}
  \acro{UEPS}{Urgency and Efficiency-based Packet Scheduling}
  \acro{UL}{uplink}
  \acro{UMTS}{Universal Mobile Telecommunications System}
  \acro{URI}{Uniform Resource Identifier}
  \acro{URM}{Unconstrained Rate Maximization}
  \acro{VR}{Virtual Resource}
  \acro{VoIP}{Voice over IP}
  \acro{WAN}{Wireless Access Network}
  \acro{WCDMA}{Wideband Code Division Multiple Access}
  \acro{WF}{Water-filling}
  \acro{WiMAX}{Worldwide Interoperability for Microwave Access}
  \acro{WINNER}{Wireless World Initiative New Radio}
  \acro{WLAN}{Wireless Local Area Network}
  \acro{WMMSE}{weighted minimum mean square error}
  \acro{WMPF}{Weighted Multicarrier Proportional Fair}
  \acro{WPF}{Weighted Proportional Fair}
  \acro{WSN}{Wireless Sensor Network}
  \acro{WWW}{World Wide Web}
  \acro{XIXO}{(Single or Multiple) Input (Single or Multiple) Output}
  \acro{ZF}{Zero-Forcing}
  \acro{ZMCSCG}{Zero Mean Circularly Symmetric Complex Gaussian}
\end{acronym}
%\end{singlespace}

\bstctlcite{IEEEexample:BSTcontrol}
%
%\onecolumn

%Possible titles

%Optimal Filter Design for Non-Coherent Over-the-Air Computation
%Over-the-Air Computation: Sample-Level Perspective
%Beating the Matched Filter

\title{Optimal Receive Filter Design for  

Misaligned Over-the-Air Computation}

\author{\IEEEauthorblockN{Henrik Hellström\IEEEauthorrefmark{1}\IEEEauthorrefmark{2}, Saeed Razavikia\IEEEauthorrefmark{1}, Viktoria Fodor\IEEEauthorrefmark{1}, Carlo Fischione\IEEEauthorrefmark{1}}
 \IEEEauthorblockA{\IEEEauthorrefmark{1}School of Electrical Engineering and Computer Science, KTH Royal Institute of Technology, Stockholm, Sweden\\
Emails: \{hhells, sraz, vjfodor, carlofi\}@kth.se}

\IEEEauthorblockA{\IEEEauthorrefmark{2}Electrical Engineering Department, Stanford University, California, USA\\
Email: hhells@stanford.edu}
}

\maketitle

\begin{abstract}	
Over-the-air computation (OAC) is a promising wireless communication method for aggregating data from many devices in dense wireless networks. The fundamental idea of OAC is to exploit signal superposition to compute functions of multiple simultaneously transmitted signals. However, the time- and phase-alignment of these superimposed signals have a significant effect on the quality of function computation. In this study, we analyze the OAC problem for a system with unknown random time delays and phase shifts. We show that the classical matched filter does not produce optimal results, and generates bias in the function estimates. To counteract this, we propose a new filter design and show that, under a bound on the maximum time delay, it is possible to achieve unbiased function computation. Additionally, we propose a Tikhonov regularization problem that produces an optimal filter given a tradeoff between the bias and noise-induced variance of the function estimates. When the time delays are long compared to the length of the transmitted pulses, our filter vastly outperforms the matched filter both in terms of bias and mean-squared error (MSE). For shorter time delays, our proposal yields similar MSE as the matched filter, while reducing the bias.
\end{abstract}

%===================================================
\section{Introduction}
%===================================================

\ac{OAC} is a non-orthogonal wireless communication method that exploits the signal-superposition property of wireless channels to compute mathematical functions \cite{gastpar2003source}. The unique selling point of \ac{OAC} is that the throughput scales linearly with the number of devices that share the multiple access channel~\cite{abari2016over}, thereby achieving significant gains for dense wireless networks. \ac{OAC} is currently an active area of research in distributed computing and has potential applications in areas such as distributed machine learning, intra-chip communications, and wireless control~\cite{csahin2023survey}. In this work, we focus on the \ac{OAC} problem itself without specifying any particular application of the method.

Although the signal-superposition property naturally computes a weighted sum of the transmitted signals, it is challenging to realize reliable function computation in practice~\cite{csahin2023survey}. One of the most notorious issues is time synchronization errors that cause the transmitted signals to arrive with time and phase misalignment at the receiving device. Consider that \ac{OAC} is taking place over the 5GHz WiFi band. Then, the period of the carrier wave is approximately $1/(5\text{GHz})=0.2\text{ns}$. If there is an unknown time synchronization error of $0.1\text{ns}$ between two transmitting devices, their superimposed signals will arrive $180^\circ$ out-of-phase, resulting in destructive interference. Meanwhile, modern standards generally cannot guarantee sub-$\mu\text{s}$ precision for IoT applications \cite{elsts2016microsecond} and sub-$260\text{ns}$ precision for 5G New Radio \cite{5gspecification}. In other words, even if the transmitting devices can perfectly compensate for the phase shifts of the wireless channel, tiny time-synchronization errors cause the superimposed signals to arrive with phase-misalignment. Note that in a point-to-point non-\ac{OAC} system, the phase of the received signal and the local oscillator can be aligned with a \ac{PLL} \cite[Chapter 7]{papananos1999radio}, but this does not solve the problem of phase alignment for the sake of superposition, as required in \ac{OAC}.

Additionally,
%given that we wish to compute more than one function using \ac{OAC}, 
time synchronization errors that exceed a few nanoseconds can cause inter-symbol interference between neighboring symbols. For instance, if the transmitting devices are communicating gradients, the elements of the gradient vectors are encoded as symbols. Generally, one wishes to pack these symbols tightly in time for the sake of communication efficiency, but with time synchronization errors, this will cause inter-symbol interference in the superimposed waveforms. In this work, we investigate the design of unbiased \ac{OAC} receivers under the real-world assumptions that no phase knowledge is available at either the transmitting or receiving devices and that the signals are misaligned in time.

%---------------------------------------------------
\subsection{Related Work}
%---------------------------------------------------

In the literature, the majority of \ac{OAC} proposals consider coherent communication schemes \cite{abari2016over, zhu2019broadband, cao2020optimized, liu2020over, zang2020over, hellstrom2022unbiased, hellstrom2023retransmission, razavikia2023computing}, i.e., schemes that rely on knowledge of the channel coefficients. All of these works consider some form of channel pre-equalization, where phase misalignment is compensated for by the transmitting devices. Implicitly, it is assumed that phase misalignment can be perfectly compensated and that there is no inter-symbol interference.

Recently, a few works have appeared that consider time synchronization errors but with supplementary information~\cite{shao2021federated, razavikia2022blind}. Specifically, \cite{shao2021federated} considers a system where rectangular pulses are transmitted with random time delays, but where the receiver knows the time delay of every device. The work in \cite{razavikia2022blind} assumes random and unknown time delays but with channel knowledge at the receiving device. They propose a synchronization-free estimator by formulating misaligned \ac{OAC} into an atomic norm minimization problem.
% =========================
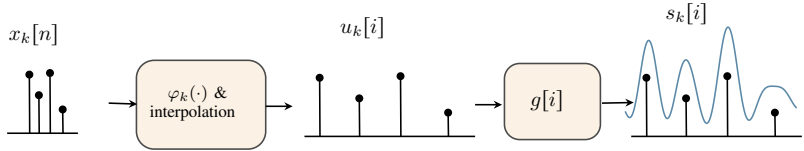
\begin{figure*}[!t]
\centering
\scalebox{0.8}{

\tikzset{every picture/.style={line width=0.75pt}} %set default line width to 0.75pt        

\begin{tikzpicture}[x=0.75pt,y=0.75pt,yscale=-1,xscale=1]
%uncomment if require: \path (0,127); %set diagram left start at 0, and has height of 127

%Straight Lines [id:da42079222359818846] 
\draw    (76.5,92) -- (121.11,92) ;
%Straight Lines [id:da9053021816658018] 
\draw [line width=0.75]    (90.71,92) -- (90.5,55) ;
\draw [shift={(90.5,55)}, rotate = 269.67] [color={rgb, 255:red, 0; green, 0; blue, 0 }  ][fill={rgb, 255:red, 0; green, 0; blue, 0 }  ][line width=0.75]      (0, 0) circle [x radius= 2.01, y radius= 2.01]   ;
% Plotting does not support converting to Tikz
%Straight Lines [id:da9570209190184586] 
\draw    (451.5,73) -- (468.5,72.67) ;
\draw [shift={(471.5,72.61)}, rotate = 178.88] [fill={rgb, 255:red, 0; green, 0; blue, 0 }  ][line width=0.08]  [draw opacity=0] (5.36,-2.57) -- (0,0) -- (5.36,2.57) -- (3.56,0) -- cycle    ;
%Straight Lines [id:da33593109297364965] 
\draw [line width=0.75]    (103.71,92) -- (103.5,54) ;
\draw [shift={(103.5,54)}, rotate = 269.68] [color={rgb, 255:red, 0; green, 0; blue, 0 }  ][fill={rgb, 255:red, 0; green, 0; blue, 0 }  ][line width=0.75]      (0, 0) circle [x radius= 2.01, y radius= 2.01]   ;
%Straight Lines [id:da8229061128387927] 
\draw [line width=0.75]    (96.71,92) -- (96.5,68) ;
\draw [shift={(96.5,68)}, rotate = 269.49] [color={rgb, 255:red, 0; green, 0; blue, 0 }  ][fill={rgb, 255:red, 0; green, 0; blue, 0 }  ][line width=0.75]      (0, 0) circle [x radius= 2.01, y radius= 2.01]   ;
%Straight Lines [id:da8110100869757357] 
\draw [line width=0.75]    (111.71,92) -- (111.5,77) ;
\draw [shift={(111.5,77)}, rotate = 269.18] [color={rgb, 255:red, 0; green, 0; blue, 0 }  ][fill={rgb, 255:red, 0; green, 0; blue, 0 }  ][line width=0.75]      (0, 0) circle [x radius= 2.01, y radius= 2.01]   ;
%Rounded Rect [id:dp6024632082236814] antiquewhite
\draw  [color={rgb, 255:red, 74; green, 74; blue, 74 }  ,draw opacity=1 ][fill={rgb, 1:red, 0.98; green, 0.92; blue, 0.84 }  ,fill opacity=0.62 ] (158,57.2) .. controls (158,51.02) and (163.01,46) .. (169.2,46) -- (228.3,46) .. controls (234.49,46) and (239.5,51.02) .. (239.5,57.2) -- (239.5,90.8) .. controls (239.5,96.99) and (234.49,102) .. (228.3,102) -- (169.2,102) .. controls (163.01,102) and (158,96.99) .. (158,90.8) -- cycle ;
%Straight Lines [id:da2998825223188297] 
\draw    (263.93,94) -- (371.5,94) ;
%Straight Lines [id:da45982805147179473] 
\draw [line width=0.75]    (273.71,94) -- (273.5,57) ;
\draw [shift={(273.5,57)}, rotate = 269.67] [color={rgb, 255:red, 0; green, 0; blue, 0 }  ][fill={rgb, 255:red, 0; green, 0; blue, 0 }  ][line width=0.75]      (0, 0) circle [x radius= 2.01, y radius= 2.01]   ;
%Straight Lines [id:da8155078784171801] 
\draw [line width=0.75]    (324.71,94) -- (324.5,56) ;
\draw [shift={(324.5,56)}, rotate = 269.68] [color={rgb, 255:red, 0; green, 0; blue, 0 }  ][fill={rgb, 255:red, 0; green, 0; blue, 0 }  ][line width=0.75]      (0, 0) circle [x radius= 2.01, y radius= 2.01]   ;
%Straight Lines [id:da25303517547750753] 
\draw [line width=0.75]    (298.71,94) -- (298.5,70) ;
\draw [shift={(298.5,70)}, rotate = 269.49] [color={rgb, 255:red, 0; green, 0; blue, 0 }  ][fill={rgb, 255:red, 0; green, 0; blue, 0 }  ][line width=0.75]      (0, 0) circle [x radius= 2.01, y radius= 2.01]   ;
%Straight Lines [id:da161882367659256] 
\draw [line width=0.75]    (354.71,94) -- (354.5,79) ;
\draw [shift={(354.5,79)}, rotate = 269.18] [color={rgb, 255:red, 0; green, 0; blue, 0 }  ][fill={rgb, 255:red, 0; green, 0; blue, 0 }  ][line width=0.75]      (0, 0) circle [x radius= 2.01, y radius= 2.01]   ;
%Straight Lines [id:da20121774864248443] 
\draw    (140.5,73) -- (155.5,73.51) ;
\draw [shift={(158.5,73.61)}, rotate = 181.93] [fill={rgb, 255:red, 0; green, 0; blue, 0 }  ][line width=0.08]  [draw opacity=0] (5.36,-2.57) -- (0,0) -- (5.36,2.57) -- (3.56,0) -- cycle    ;
%Straight Lines [id:da19477253615657153] 
\draw    (239.5,75) -- (253.5,75) ;
\draw [shift={(256.5,75)}, rotate = 180] [fill={rgb, 255:red, 0; green, 0; blue, 0 }  ][line width=0.08]  [draw opacity=0] (5.36,-2.57) -- (0,0) -- (5.36,2.57) -- (3.56,0) -- cycle    ;
%Rounded Rect [id:dp04485605703492013] 
\draw  [color={rgb, 255:red, 74; green, 74; blue, 74 }  ,draw opacity=1 ][fill={rgb, 1:red, 0.98; green, 0.92; blue, 0.84 }  ,fill opacity=0.62 ] (390,57.6) .. controls (390,52.3) and (394.3,48) .. (399.6,48) -- (440.9,48) .. controls (446.2,48) and (450.5,52.3) .. (450.5,57.6) -- (450.5,86.4) .. controls (450.5,91.7) and (446.2,96) .. (440.9,96) -- (399.6,96) .. controls (394.3,96) and (390,91.7) .. (390,86.4) -- cycle ;
%Straight Lines [id:da20110638294867056] 
\draw    (371.5,74) -- (385.5,74) ;
\draw [shift={(388.5,74)}, rotate = 180] [fill={rgb, 255:red, 0; green, 0; blue, 0 }  ][line width=0.08]  [draw opacity=0] (5.36,-2.57) -- (0,0) -- (5.36,2.57) -- (3.56,0) -- cycle    ;
%Straight Lines [id:da4479901425563695] 
\draw    (469.93,94) -- (577.5,94) ;
%Straight Lines [id:da20715876305785685] 
\draw [line width=0.75]    (479.71,94) -- (479.5,57) ;
\draw [shift={(479.5,57)}, rotate = 269.67] [color={rgb, 255:red, 0; green, 0; blue, 0 }  ][fill={rgb, 255:red, 0; green, 0; blue, 0 }  ][line width=0.75]      (0, 0) circle [x radius= 2.01, y radius= 2.01]   ;
%Straight Lines [id:da46356513399808486] 
\draw [line width=0.75]    (530.71,94) -- (530.5,56) ;
\draw [shift={(530.5,56)}, rotate = 269.68] [color={rgb, 255:red, 0; green, 0; blue, 0 }  ][fill={rgb, 255:red, 0; green, 0; blue, 0 }  ][line width=0.75]      (0, 0) circle [x radius= 2.01, y radius= 2.01]   ;
%Straight Lines [id:da6485581275258681] 
\draw [line width=0.75]    (504.71,94) -- (504.5,70) ;
\draw [shift={(504.5,70)}, rotate = 269.49] [color={rgb, 255:red, 0; green, 0; blue, 0 }  ][fill={rgb, 255:red, 0; green, 0; blue, 0 }  ][line width=0.75]      (0, 0) circle [x radius= 2.01, y radius= 2.01]   ;
%Straight Lines [id:da7329433462756967] 
\draw [line width=0.75]    (560.71,94) -- (560.5,79) ;
\draw [shift={(560.5,79)}, rotate = 269.18] [color={rgb, 255:red, 0; green, 0; blue, 0 }  ][fill={rgb, 255:red, 0; green, 0; blue, 0 }  ][line width=0.75]      (0, 0) circle [x radius= 2.01, y radius= 2.01]   ;

% Text Node
\draw (76,21.4) node [anchor=north west][inner sep=0.75pt]  [font=\normalsize]  {$x_{k}[n]$};
% Text Node
\draw (165,60.2) node [anchor=north west][inner sep=0.75pt]  [font=\footnotesize] [align=left] {$~~~\varphi_k(\cdot)$ \& \\interpolation};
% Text Node
\draw (285,17.4) node [anchor=north west][inner sep=0.75pt]  [font=\normalsize]  {$u_{k}[i]$};
% Text Node
\draw (405,63.4) node [anchor=north west][inner sep=0.75pt]  [font=\normalsize]  {$g[i]$};
% Text Node
\draw (490,8.4) node [anchor=north west][inner sep=0.75pt]  [font=\normalsize]  {$s_{k}[i]$};

\begin{axis}[
width=5cm,
height=3.5cm,
xticklabels=none,
 yticklabels=none,
 xtick=\empty,
 ytick=\empty,
 axis line style={draw=none},
 xshift=12cm,yshift=0.5cm,
 ]
\addplot[
    domain=-20:20,
    samples=500, 
    color=airforceblue,
]{-15*sin(deg(x+14.5))/(x+14.5)-10*sin(deg(x+6))/(x+6)-6*sin(deg(x-15.1))/(x-15.1) -20*sin(deg(x-4))/(x-4)};
\end{axis}
% \begin{axis}[
% width=5cm,
% height=3.5cm,
% xticklabels=none,
%  yticklabels=none,
%  xtick=\empty,
%  ytick=\empty,
%  axis line style={draw=none},
%  xshift=1cm,yshift=0.5cm,
%  ]
% \addplot[
%     domain=0:2,
%     samples=500, 
%     color=airforceblue,
% ]{1 + 1*cos(deg(3*x))};
% \end{axis}

\end{tikzpicture}
}
\caption{
Transmitter-side baseband illustration. Here,  $x_k[n]$ denotes the messages of device $k$, which are encoded into IQ symbols by the function $\varphi_k$ and interpolated by $N_s$ samples to generate the pulse train $u_k[i]$. Then, the baseband waveform is generated by convolving with the pulse-shaping filter $g[i]$. Finally, in a practical system, the signal $s_k[i]$ would be upconverted to the carrier frequency and transmitted over the \ac{MAC}. However, in this paper, we only model the baseband representation.
}
\label{fig:Transmitter}
\end{figure*}

% =========================

Finally, there are some proposals that consider both time misalignment and unknown phases of the channel coefficients \cite{goldenbaum2009function, goldenbaum2013robust, kortke2014analog, csahin2022over, csahin2023distributed, csahin2022demonstration}. These works can be divided into two groups, majority vote \ac{OAC} \cite{csahin2022over, csahin2023distributed, csahin2022demonstration} and analog amplitude-modulated \ac{OAC} \cite{goldenbaum2009function, goldenbaum2013robust, kortke2014analog}. In majority vote \ac{OAC}, orthogonal radio resources are allocated for every symbol, and the devices vote by transmitting in one of the orthogonal resources. In the analog amplitude-modulated \ac{OAC} scheme, each device transmits a unimodular sequence of random phases at a transmit power that depends on the transmitted message. This results in a good estimator, at the cost of reducing communication efficiency proportionally to the length of the transmitted sequence. Both the analog amplitude-modulated \ac{OAC} and the majority vote \ac{OAC} has been implemented in practice using software-defined radio \cite{kortke2014analog, csahin2022demonstration}.

Our work can be viewed as an extension of analog amplitude-modulated \ac{OAC} where we consider a sample-level model. This is in contrast to what is done in \cite{goldenbaum2013robust}, where the authors present a symbol-level \ac{OAC} scheme that describes a level of abstraction at which the symbol is the smallest unit of information; such a description abstracts away that each symbol is actually comprised of multiple samples, which form the baseband waveform. By presenting a sample-level scheme, our work reveals multiple features of \ac{OAC} that cannot be expressed analytically using symbols, such as handling sub-symbol level time offsets.
%---------------------------------------------------
\subsection{Contributions}
%---------------------------------------------------
\begin{itemize}
    \item \textbf{Unbiased Time-Asynchronous \ac{OAC}}: We show that, given an upper bound on the time synchronization error, it is possible to achieve unbiased function estimation with \ac{OAC}, even without phase knowledge or knowledge of the time delays. We provide the conditions under which this occurs, relating to the maximum delay, the number of samples per symbol, and the pulse shaping filter.
    \item \textbf{Receive Filter Design}: We demonstrate that the matched filter is not optimal for time misaligned \ac{OAC} both in terms of bias and \ac{MSE}. Therefore we propose a new receive filter that yields completely unbiased function estimation. However, our unbiased filter amplifies the noise considerably. Therefore we also propose an optimal filter design that relaxes the bias constraint to significantly reduce the noise amplification. This optimal filter significantly outperforms the matched filter, both in terms of bias and \ac{MSE}.
\end{itemize}
%---------------------------------------------------
%\subsection{Organization of the paper}
%---------------------------------------------------

The rest of the paper is organized as follows: Section \ref{sec:standard} describes the sample-level baseband for a general receive filter $\bm{A}$. Section \ref{sec:as_OAC} derives the conditions on $\bm{A}$ to achieve unbiased function estimation and develops our proposed filter. Section \ref{sec:num} contains numerical results that compare the performance of our filter to the matched filter. Finally, Section \ref{sec:conclusion} concludes the paper.

%---------------------------------------------------
\subsection{Notation}
%---------------------------------------------------

Throughout this paper, scalars are denoted by lowercase letters $x$, vectors and matrices by lower $\bm{x}$, and upper-case boldface letters $\bm{X}$, respectively. Operators are represented by calligraphic notations such as $\mathcal{X}$. The transpose, Hermitian, and the element-wise complex conjugate of a matrix $\bm{X}$ are represented by $\bm{X}^{\mathsf{T}}$, $\bm{X}^{\mathsf{H}}$, and $\overline{\bm{X}}$, respectively. We further use $\otimes$ to show the Kronecker product and $\odot$ to show the Hadamard product. Moreover, $\bm{e}_n\in \mathbb{R}^{N}$ stands for $n$-th canonical basis vector whose $n$-th element is one and zeros elsewhere.

The Hankel lifting operator $\mathscr{H}:\mathbb{R}^{L} \mapsto  \mathbb{R}^{(d+1)\times (L-d)}$ for $d\in [L-1]$ is defined as follows
%---------------
\begin{equation}
    \label{eq:Hankel}
    \mathscr{H}(\bm{x}): = \begin{bmatrix}
        x_{d+1} & x_{d+2} & \dots & x_{L} \\
        x_{d} & x_{d+1} & \dots & x_{L-1} \\
        \vdots \\
        x_{1} & x_{2} & \dots & x_{L-d} 
    \end{bmatrix}.
\end{equation}
%---------------

%===================================================
\section{System Model}\label{sec:standard}
%===================================================

In this section, we describe a system model for \ac{OAC} at the sample level. We describe all components of the baseband transceiver and the wireless channel model, including synchronization errors in discrete time.

%---------------------------------------------------
\subsection{Transmitter Architecture}
%---------------------------------------------------

Consider a communication network with $K$ devices and one \ac{FC}, where device $k$ carries a sequence of $N$ messages $x_k[n]$ taking values from a compact subset of real domain, i.e., $\mathbb{D}_f := [x_{\textnormal{min}}, x_{\textnormal{max}}]\subset \mathbb{R}$ for discrete $n \in [N]$ and $k\in [K]$.  We aim to compute functions  $f_n(x_1[n],x_2[n],\ldots,x_K[n])$ at the \ac{FC} via the wireless channel for $n \in [N]$. Without loss of generality, we consider $f$ to be the arithmetic mean, i.e.,  
 %--------------
 \begin{align}
    \label{eq:function}
     f_n = \frac{1}{K}\sum\nolimits_{k=1}^Kx_k[n], \quad \forall~n\in [N].
 \end{align}
 %--------------
Note that \ac{OAC} theoretically can compute all nomographic functions by adding pre- and post-processing to the transmitted and received signals, respectively~\cite{Golden2013Harnessing}. 

To transmit the messages $x_k[n]$,  the values are encoded in \ac{IQ} symbols according to the analog modulation function $\varphi_k: \mathbb{D}\mapsto \mathbb{R}$. As in \cite{goldenbaum2013robust}, we consider an \ac{AM} scheme where the magnitude of the uplink channel coefficients $h_k$ are known. Specifically, the square root of the messages $x_k[n]$ are encoded with the modulation function $\varphi_k(x_k[n]) = \sqrt{x_k[n]}/|h_k|$\footnote{Some devices may have insufficient transmission power to realize $1/|h_k|$. However, we consider that problem out-of-scope and refer to \cite{cao2020optimized, liu2020over, zang2020over, hellstrom2022unbiased, hellstrom2023retransmission}.}. After modulation, each symbol $\varphi_k(x_k[n])$ is convolved with a pulse shaping filter $g[i]$ of size $N_s$, where $n$ indexes symbols and $i$ indexes samples. As a result, each symbol $\varphi_k(x_k[n])$ is transmitted with a pulse consisting of $N_s$ samples. Mathematically, the baseband signal is given by
 %--------------
\begin{equation}
\label{eq:interpolation}
    s_k[i] = g[i]\circledast   u_k[i], 
\end{equation}
 %--------------
where $u_k[i]$ is the interpolated signal of $\varphi_k(x_k[n])$ by the integer factor $N_s$, i.e., 
  %--------------
\begin{equation}
\label{eq:upsample}
    u_k[i] = \begin{cases}
        \varphi_k(x_k[i/N_s]), & N_s~{\rm divides}~i \\
        0, & {\rm otherwise}.
    \end{cases}
\end{equation}
 %--------------
To simplify the notation, we use the following matrix representation.
%--------------
\begin{align}
    \bm{s}_k =\bm{G}\bm{\varphi}_k,
\end{align}
%--------------
where  $\bm{\varphi}_k = [\varphi_k(x_k[1]), \ldots, \varphi_k(x_k[N])]^{\mathsf{T}} \in \mathbb{R}^{N}$ and $\bm{G} : = \bm{I}_N \otimes \bm{g} \in \mathbb{R}^{N N_s \times N}$ is the pulse shaping matrix where $\bm{g} : = [g[1], \ldots, g[N_s]]^{\mathsf{T}}$ are the filter taps $g[i]$. This baseband signal $\bm{s}_k$ is a discrete sequence of $N_t := N\times N_s$ samples, that convey the information of the $N$ symbols $\varphi_k(x_k[N])$. Fig~\ref{fig:Transmitter} depicts The procedure of producing the baseband signal. 

As suggested in \cite{goldenbaum2013robust}, device $k$ does not only transmit this sequence $\bm{s}_k$ once but generates an $M$-length unimodular sequence of random phases to communicate the functions. In our scheme, this can be realized by communicating $M$ copies of $\bm{s}_k$, phase-shifted by $\theta_{k,m}\sim\mathcal{U}[0,2\pi)$, resulting in
\begin{align}\label{eq:tx_m}
    \bm{s}_k^{(m)} =\bm{G}\bm{\varphi}_k{\rm e}^{j\theta_{k,m}}, \quad \forall~m\in [M].
\end{align}
For conciseness, we will only invoke the superscript $(m)$ as deemed necessary.

%=============================
\subsection{Modeling Time offset of the channel}
%=============================

As a result of time asynchronous transmission in the uplink (from the devices to the \ac{FC}), the pulse-shaped baseband signals $s_k[i]$ are shifted in time with respect to one another. Therefore, the \ac{FC} receives the superposition of time-shifted signals $s_k[i-d_k]$. For simplicity's sake, we assume that the $K$ signals are shifted with integer multiples of the sampling time $T_s=1/f_s$. Device $k$ is subject to a delay of $d_k \leq d < N_s$ samples, i.e., symbol-level synchronization. These shifts are expressed by the shifting matrix $\bm{E} \in \mathbb{R}^{N_t\times N_t}$ where $N_t := N\times N_s$ as follows
%--------------
\begin{equation}
    \label{eq:timeshift}
    \bm{E} : = \begin{bmatrix}
            \bm{0}_{N_t-1}^{\mathsf{T}} & 0 \\
            \bm{I}_{N_t-1} &  \bm{0}_{N_t-1}
    \end{bmatrix},
\end{equation}
%--------------
where $\bm{0}_{N_t-1}$ is the zero vector of size $N_t-1$. For a signal shifted by $d_k$ samples, the pulse shape is multiplied by $\bm{E}^{d_k}\bm{G}$. 

% =======================
\begin{figure*}
    \centering
\scalebox{0.75}{

\tikzset{every picture/.style={line width=0.75pt}} %set default line width to 0.75pt        

\begin{tikzpicture}[x=0.75pt,y=0.75pt,yscale=-1,xscale=1]
%uncomment if require: \path (0,217); %set diagram left start at 0, and has height of 217

%Straight Lines [id:da42079222359818846] 
\draw    (82.5,74) -- (203.5,74) ;
%Straight Lines [id:da9570209190184586] 
\draw    (537.5,123) -- (554.5,122.67) ;
\draw [shift={(557.5,122.61)}, rotate = 178.88] [fill={rgb, 255:red, 0; green, 0; blue, 0 }  ][line width=0.08]  [draw opacity=0] (5.36,-2.57) -- (0,0) -- (5.36,2.57) -- (3.56,0) -- cycle    ;
%Straight Lines [id:da2998825223188297] 
\draw    (364.93,144) -- (472.5,144) ;
%Straight Lines [id:da19477253615657153] 
\draw    (288.5,124.5) -- (307.5,124.07) ;
\draw [shift={(310.5,124)}, rotate = 178.7] [fill={rgb, 255:red, 0; green, 0; blue, 0 }  ][line width=0.08]  [draw opacity=0] (5.36,-2.57) -- (0,0) -- (5.36,2.57) -- (3.56,0) -- cycle    ;
%Rounded Rect [id:dp04485605703492013] 
\draw  [color={rgb, 255:red, 74; green, 74; blue, 74 }  ,draw opacity=1 ][fill={rgb, 1:red, 0.98; green, 0.92; blue, 0.84 }  ,fill opacity=0.62 ] (501,110.3) .. controls (501,106.27) and (504.27,103) .. (508.3,103) -- (530.2,103) .. controls (534.23,103) and (537.5,106.27) .. (537.5,110.3) -- (537.5,132.7) .. controls (537.5,136.73) and (534.23,140) .. (530.2,140) -- (508.3,140) .. controls (504.27,140) and (501,136.73) .. (501,132.7) -- cycle ;
%Straight Lines [id:da20110638294867056] 
\draw    (478.5,123) -- (497.5,123) ;
\draw [shift={(500.5,123)}, rotate = 180] [fill={rgb, 255:red, 0; green, 0; blue, 0 }  ][line width=0.08]  [draw opacity=0] (5.36,-2.57) -- (0,0) -- (5.36,2.57) -- (3.56,0) -- cycle    ;
% Plotting does not support converting to Tikz
%Straight Lines [id:da3909798428771629] 
\draw    (87.5,208) -- (208.5,208) ;
% Plotting does not support converting to Tikz
%Straight Lines [id:da1641362887169251] 
\draw    (224.5,70) -- (266.43,113.83) ;
\draw [shift={(268.5,116)}, rotate = 226.27] [fill={rgb, 255:red, 0; green, 0; blue, 0 }  ][line width=0.08]  [draw opacity=0] (5.36,-2.57) -- (0,0) -- (5.36,2.57) -- (3.56,0) -- cycle    ;
%Straight Lines [id:da7194484935328394] 
\draw    (221.5,175) -- (267.22,135.95) ;
\draw [shift={(269.5,134)}, rotate = 139.5] [fill={rgb, 255:red, 0; green, 0; blue, 0 }  ][line width=0.08]  [draw opacity=0] (5.36,-2.57) -- (0,0) -- (5.36,2.57) -- (3.56,0) -- cycle    ;
%Flowchart: Or [id:dp7504131410395087] 
\draw   (265,124.5) .. controls (265,118.15) and (270.26,113) .. (276.75,113) .. controls (283.24,113) and (288.5,118.15) .. (288.5,124.5) .. controls (288.5,130.85) and (283.24,136) .. (276.75,136) .. controls (270.26,136) and (265,130.85) .. (265,124.5) -- cycle ; \draw   (265,124.5) -- (288.5,124.5) ; \draw   (276.75,113) -- (276.75,136) ;
% Plotting does not support converting to Tikz
%Flowchart: Or [id:dp198768969885454] 
\draw   (310.5,124) .. controls (310.5,117.65) and (315.76,112.5) .. (322.25,112.5) .. controls (328.74,112.5) and (334,117.65) .. (334,124) .. controls (334,130.35) and (328.74,135.5) .. (322.25,135.5) .. controls (315.76,135.5) and (310.5,130.35) .. (310.5,124) -- cycle ; \draw   (310.5,124) -- (334,124) ; \draw   (322.25,112.5) -- (322.25,135.5) ;
%Straight Lines [id:da7141958624656499] 
\draw    (334,124) -- (353,123.57) ;
\draw [shift={(356,123.5)}, rotate = 178.7] [fill={rgb, 255:red, 0; green, 0; blue, 0 }  ][line width=0.08]  [draw opacity=0] (5.36,-2.57) -- (0,0) -- (5.36,2.57) -- (3.56,0) -- cycle    ;
%Straight Lines [id:da8476395110499635] 
\draw    (322.5,93) -- (322.29,109.5) ;
\draw [shift={(322.25,112.5)}, rotate = 270.73] [fill={rgb, 255:red, 0; green, 0; blue, 0 }  ][line width=0.08]  [draw opacity=0] (5.36,-2.57) -- (0,0) -- (5.36,2.57) -- (3.56,0) -- cycle    ;
%Rounded Rect [id:dp7045108827691717] 
\draw  [color={rgb, 255:red, 74; green, 74; blue, 74 }  ,draw opacity=1 ][fill={rgb, 1:red, 0.98; green, 0.92; blue, 0.84 }  ,fill opacity=0.62 ] (559,111.3) .. controls (559,107.27) and (562.27,104) .. (566.3,104) -- (588.2,104) .. controls (592.23,104) and (595.5,107.27) .. (595.5,111.3) -- (595.5,133.7) .. controls (595.5,137.73) and (592.23,141) .. (588.2,141) -- (566.3,141) .. controls (562.27,141) and (559,137.73) .. (559,133.7) -- cycle ;
%Straight Lines [id:da21044842368935668] 
\draw    (595.5,123) -- (612.5,122.67) ;
\draw [shift={(615.5,122.61)}, rotate = 178.88] [fill={rgb, 255:red, 0; green, 0; blue, 0 }  ][line width=0.08]  [draw opacity=0] (5.36,-2.57) -- (0,0) -- (5.36,2.57) -- (3.56,0) -- cycle    ;

% Text Node
\draw (100,8.4) node [anchor=north west][inner sep=0.75pt]  [font=\normalsize]  {$\bm{E}^{d_{1}}\bm{s}_{1}$};
% Text Node
\draw (400,62.4) node [anchor=north west][inner sep=0.75pt]  [font=\normalsize]  {$|\bm{v}|$};
% Text Node
\draw (511,112.7) node [anchor=north west][inner sep=0.75pt]  [font=\normalsize]  {$\bm{A}$};
% Text Node
\draw (100,142.4) node [anchor=north west][inner sep=0.75pt]  [font=\normalsize]  {$\bm{E}^{d_{K}}\bm{s}_{K}$};
% Text Node
\draw (134,99.4) node [anchor=north west][inner sep=0.75pt]  [font=\LARGE]  {$\vdots $};
% Text Node
\draw (245,65.4) node [anchor=north west][inner sep=0.75pt]  [font=\normalsize]  {$h_{1}$};
% Text Node
\draw (247.5,157.9) node [anchor=north west][inner sep=0.75pt]  [font=\normalsize]  {$h_{K}$};
% Text Node
\draw (318,69.4) node [anchor=north west][inner sep=0.75pt]  [font=\normalsize]  {$\bm{z}$};
% Text Node
\draw (544,90.4) node [anchor=north west][inner sep=0.75pt]  [font=\normalsize]  {$\bm{y}$};
% Text Node
\draw (565,113.4) node [anchor=north west][inner sep=0.75pt]  [font=\normalsize]  {$\psi(\cdot)$};
% Text Node
\draw (621,109.4) node [anchor=north west][inner sep=0.75pt]  [font=\normalsize]  {$\hat{\bm{f}}$};

% = ===== Curve ===== g1(t) 

\begin{axis}[
width=5cm,
height=3.5cm,
xticklabels=none,
 yticklabels=none,
 xtick=\empty,
 ytick=\empty,
 axis line style={draw=none},
 xshift=2.25cm,yshift=-0.1cm,
 ]
\addplot[
    domain=-20:20,
    samples=500, 
    color=cadmiumorange,
]{0.1*(-5*sin(2*deg(x))/x+cos(deg(x)-3))};
\end{axis}

% = ===== Curve ===== gr(t)
\begin{axis}[
width=5cm,
height=3.5cm,
xticklabels=none,
 yticklabels=none,
 xtick=\empty,
 ytick=\empty,
 axis line style={draw=none},
 xshift=2.25cm,yshift=3.55cm,
 ]
\addplot[
    domain=-20:20,
    samples=500, 
    color=cadmiumgreen,
]{0.1*(-5*sin(2*deg(x-5))/(x-5)+cos(deg(x)-1))};
\end{axis}
% = ===== Curve ===== y(t)
\begin{axis}[
width=5cm,
height=3.5cm,
xticklabels=none,
 yticklabels=none,
 xtick=\empty,
 ytick=\empty,
 axis line style={draw=none},
 xshift=9.25cm,yshift=1.8cm,
 ]
\addplot[
    domain=-20:20,
    samples=500, 
    color=airforceblue,
]{-15*sin(deg(x+14))/(x+14)+cos(deg(x)-1)-10*sin(deg(x-7))/(x-7)+cos(deg(x)-1) -20*sin(deg(x-13))/(x-13) };
\end{axis}

\end{tikzpicture}

}

    \caption{The schematic of our asynchronous OAC model. Here, device $k$ transmits the waveform $\bm{s}_k$ over the MAC and the \ac{FC} receives $\bm{v} = \sum_{k=1}^{K}\bm{E}^{d_k}h_k\bm{s}_k + \bm{z}$ where  $\bm{E}^{d_k}$ and $h_k$ model the time synchronization error and channel coefficient for device $k$, respectively. $\bm{z}$ denotes \ac{AWGN} at the receiver. After \ac{ADC}, the received signal $\bm{v}$ is filtered by $\bm{A}$ followed by the function $\psi(\cdot)$ to estimate the desired function. }
    \label{fig:receiver}
\end{figure*}
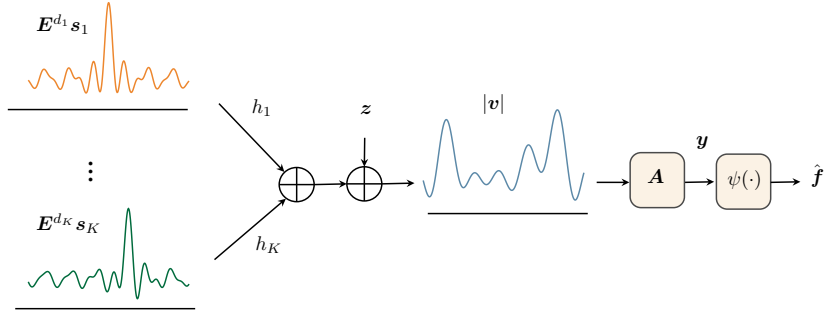
% =======================

%=============================
\subsection{Receiver Architecture}\label{sec:rx_arch}
%=============================

Invoking the given notation, the superimposed samples given by the wireless \ac{MAC} are 
%--------------
\begin{equation}
    \bm{v} = \sum\nolimits_{k=1}^Kh_k\bm{E}^{d_k}\bm{s}_k + \bm{z}.
\end{equation}
%--------------
where $\bm{z}\in\mathbb{C}^{N_t}\sim\mathcal{N}(\bm{0}, \bm{I}\sigma_z^2)$ is \ac{AWGN} and $h_k\in\mathbb{C}$ is the uplink fading coefficient from device $k$ to the \ac{FC}. For notational simplicity, $h_k$ is assumed flat for $N_t$ samples, but our scheme works as long as it is flat for $N_s$ samples. We also consider that device $k$ is able to estimate $|h_k|$ perfectly but that the phase $\angle h_k$ is unknown and uniformly distributed, i.e., $\angle h_k\sim \mathcal{U}[0,2\pi)$, as motivated in the introduction. Our channel model does not contain a path loss factor, since we assume large-scale fading is compensated via power control. In a standard communication system, the receiver would apply a matched filter and decimate the output by $N_s$ samples to recover the symbol vector. However, due to the asynchronous channel, it is not obvious that the matched filter is the best approach. Therefore, we consider a general receive filter as $\bm{A}\in\mathbb{C}^{N\times N_t}$, which is applied at the receiver side to generate the filtered signal as
%--------------
\begin{equation}
    \label{eq:received}
    \bm{y} := \bm{A}\bm{v} = \bm{A}\sum\nolimits_{k=1}^Ke^{j\theta_k}\bm{E}^{d_k}\bm{G}\sqrt{\bm{x}_k} + \tilde{\bm{z}},
\end{equation}
%--------------
 where $\tilde{\bm{z}}$ is channel noise with $\mathcal{N}(\bm{0}, \bm{A}\sigma_z^2)$ and $\theta_k$ is the phase of $h_k$. Finally, the receiver exploits the $M$ phase-shifted copies of $\bm{y}$ to form the estimate of the desired function. Specifically, the receiver follows the symbol-level procedure of \cite{goldenbaum2013robust} expressed using the function~$\psi(\cdot)$ as
%--------------
\begin{equation}
    \label{eq:fMeasure}
    \hat{\bm{f}}' = \psi\left(\bm{y}^{(1)},...,\bm{y}^{(M)}\right) = \frac{1}{MK}\sum\nolimits_{m=1}^M\overline{\bm{y}}^{(m)}\odot\bm{y}^{(m)},
\end{equation}
where
\begin{equation}
    \bm{y}^{(m)} = \bm{A}\sum\nolimits_{k=1}^Ke^{j(\tilde{\theta}_m^k)}\bm{E}^{d_k}\bm{G}\sqrt{\bm{x}_k} + \tilde{\bm{z}}^{(m)},
\end{equation}
%--------------
$\overline{\bm{y}}$ is the complex-conjugate of $\bm{y}$, $\odot$ is the Hadamard product, and $\tilde{\theta}_m^k: = \theta_{k,m}+\theta_k$  has uniform distribution, i.e., $ \tilde{\theta}_m^k \sim \mathcal{U}[0,2\pi)$. The overall communication procedure is summarized in Fig \ref{fig:receiver}.

For the $n$-th element of vector $\hat{\bm{f}}'$, we have
%--------------
\begin{align}
    \hat{f}_n'  =  \frac{1}{MK}\sum\nolimits_{m=1}^M\Big|\sum\nolimits_{k=1}^K {\rm e}^{j\tilde{\theta}_m^k}\bm{a}_n^{\mathsf{T}}\bm{B}_k\sqrt{\bm{x}_k}+\tilde{\bm{z}}^{(m)}\Big|^2,   
\end{align}
%--------------
for $n \in [N]$, where $\bm{a}_n$ is the $n$-th row of matrix $\bm{A}$ and $\bm{B}_k := \bm{E}^{d_k}\bm{G}\in \mathbb{R}^{N_t\times N}$. We further have 
%--------------
\begin{align}
    \hat{f}_n' = \frac{1}{K}\left(\sum\nolimits_{k=1}^K |\bm{a}_n^{\mathsf{T}}\bm{B}_k\sqrt{\bm{x}_k}|^2 + r_M + z_M\right),
\end{align}
%--------------
where $r_M$ involves the cross-correlation terms, i.e.,
%--------------
\begin{align*}
    r_M: = 2\operatorname{Re}\bigg[\sum_{\myoverset{k}{k'\neq k}}^K\sum_{m=1}^M\frac{{\rm e}^{j(\tilde{\theta}_m^k-\tilde{\theta}_m^{k'})}}{M}\bm{a}_n^{\mathsf{T}}\bm{B}_k\sqrt{\bm{x}_k}\sqrt{\bm{x}_{k'}}^{\mathsf{H}}\bm{B}_{k'}^{\mathsf{T}}\bm{a}_n\bigg],
\end{align*}
%--------------
and $z_M$ involves all the noise-related terms, i.e.,
%--------------
\begin{align}
    \nonumber
    % \label{eq:noise_error}
    z_M: = \frac{1}{M}\sum\nolimits_{m=1}^M\left(2\operatorname{Re}\left[\tilde{\bm{z}}\sum\nolimits_{k=1}^K {\rm e}^{j\tilde{\theta}_m^k}\bm{a}_n^{\mathsf{T}}\bm{B}_k\sqrt{\bm{x}_k}\right] + |\tilde{\bm{z}}|^2\right).
\end{align}
%--------------
In expectation, $\mathbb{E}_{\tilde{\bm{\theta}}}[r_M]=0$ and $\mathbb{E}_{\bm{z}}[z_M] = |\bm{a_n}|^2\sigma_z^2$. Thus, the receiver subtracts this term from $\hat{f}_n'$ to form the final estimate 
%--------------
\begin{align}\label{eq:final_est}
    \hat{f}_n = \hat{f}_n' - |\bm{a_n}|^2\sigma_z^2/K.
\end{align}
%--------------
For a detailed analysis of the error terms $r_M$ and $z_M$, we refer to \cite{goldenbaum2013robust}. In expectation, we have
%--------------
\begin{align}
    \label{eq:estmiedExp}
    \mathbb{E}_{\tilde{\bm{\theta}}, \bm{z}} [\hat{f}_n] = \frac{1}{K}\bm{a}_n^{\mathsf{T}}\Big (\sum\nolimits_{k=1}^K \bm{B}_k\sqrt{\bm{x}_k}  \sqrt{\bm{x}_k}^{\mathsf{H}} \bm{B}_k^{\mathsf{T}}\Big) \bm{a}_n.
\end{align}
%--------------
Now, if there is no synchronization error, then $d_k=0$ and $\bm{B}_k = \bm{G}$ for $k\in [K]$. Hence, by setting $\bm{A} = \bm{G}^{\mathsf{T}}$, \eqref{eq:estmiedExp} yields 
%--------------
\begin{align}
    \nonumber
    &\mathbb{E}_{\tilde{\bm{\theta}}, \bm{z}} [\hat{f}_n] = \frac{1}{K}\bm{g}_n^{\mathsf{T}}\bm{G}\Big(\sum\nolimits_{k=1}^K \sqrt{\bm{x}_k}\sqrt{\bm{x}_k}^{\mathsf{T}} \Big) \bm{G}^{\mathsf{T}}\bm{g}_n \\ \label{eq:estmiedfn}
    &=\frac{1}{K}\bm{e}_n^{\mathsf{T}}\Big(\sum\nolimits_{k=1}^K \sqrt{\bm{x}_k}\sqrt{\bm{x}_k}^{\mathsf{T}} \Big) \bm{e}_n  = \frac{1}{K}\sum\nolimits_{k=1}^K {x}_k[n],
\end{align}
%--------------
where $\bm{g}_n$ denotes the $n$-th column of matrix $\bm{G}$ and recall that $\bm{e}_n\in \mathbb{R}^{N}$ is $n$-th canonical basis. This tells us that if there is no synchronization error, the classical matched filter $\bm{A} = \bm{G}^{\mathsf{T}}$ generates unbiased estimates of the desired functions for the $N$ transmitted symbols.

However, for nonzero delays $d_k$, $\bm{B}_k \neq \bm{G}$ and the matched filter no longer yields unbiased estimates. Therefore, we wish to find $\bm{A}$ such that $\hat{\bm{f}}$ is an unbiased estimator for the desired functions $\bm{f}:=[f_1,\ldots,f_N]^{\mathsf{T}}$, where $f_n$ defined in \eqref{eq:function}. This is a satisfiability problem that can be expressed as
%--------------
\begin{align}
    \label{eq:ProblemP1}
     \mathcal{P}_1 ={\rm find} \quad \bm{A},
    \quad   {\rm s.t.} \quad \mathbb{E}_{\tilde{\bm{\theta}}, \bm{z}} \left[ \hat{\bm{f}}(\bm{y}) \right] = \frac{1}{K}\sum\nolimits_{k=1}^K\bm{x}_k.
\end{align}
%--------------
If such an $\bm{A}$ exists, we want an algorithm that finds it without knowledge of any $d_k$ but with knowledge of $d$, i.e., the maximum delay. In the next section, we investigate the condition under which $\bm{A}$ can give an unbiased estimator. 

%===================================================
\section{Asynchronous OAC}\label{sec:as_OAC}
%===================================================

In this section, we first seek the conditions for the receive filter $\bm{A}$ to obtain an unbiased estimator. Then, we propose a design for $\bm{A}$ that minimizes channel noise while retaining a near-zero bias.

To cancel the effect of the time-shifting matrices $\bm{B}_k$ in \eqref{eq:estmiedExp},  we propose finding $\bm{a}_n$ such that the matrix multiplications of $\bm{B}_k$ with $\bm{a}_n$ becomes $\bm{e}_n$ for all $k$, i.e., 
%--------------
\begin{align}
    \label{eq:a_nprod}
    \bm{a}_n^{\mathsf{T}}\bm{B}_k = \bm{e}_n^{\mathsf{T}}, \quad k\in [K],~~ n\in [N].
\end{align}
%--------------
Note that such a choice solves \eqref{eq:ProblemP1} exactly but restricts the solution to a subset of the feasible set. Recall that $\bm{B}_k = \bm{E}^{d_k}\bm{G}$. Since the receiver does not know $d_k$, the filter $\bm{a}_n$ must be designed such that \eqref{eq:a_nprod} holds for all possible delays $d_k\in[0,d]$. Since matrix $\bm{G}$ has a block diagonal structure ($\bm{G} : = \bm{I}_N \otimes \bm{g}$), matrix $\bm{B}_k$ also becomes block diagonal. Therefore, only  $N_s$  elements of $\bm{a}_n$ contributes to the multiplication in \eqref{eq:a_nprod}. In particular, $\bm{a}_n$ has the structure
 %--------------
 \begin{align}
    \label{eq:a_ne_n}
     \bm{a}_n = \bm{e}_n \otimes \tilde{\bm{a}},
 \end{align}
 %--------------
where $\tilde{\bm{a}} \in \mathbb{R}^{N_s}$. Moreover, to characterize the inter-symbol interference error, we decompose vector $\tilde{\bm{a}}$ into two parts as $\tilde{\bm{a}} = \tilde{\bm{a}}_1 + \tilde{\bm{a}}_2$, where
%--------------
\begin{subequations}
\label{eq:alphatotal}
\begin{align}
    \label{eq:alpha1}
    \tilde{\bm{a}}_1 &  = [\bm{\beta}^{\mathsf{T}},\bm{0}_{N_s-d}^{\mathsf{T}}]^{\mathsf{T}}, \quad \bm{\beta} \in \mathbb{R}^{d} \\\label{eq:alpha2}
    \tilde{\bm{a}}_2 &  = [\bm{0}_d^{\mathsf{T}}, \bm{\alpha}^{\mathsf{T}}]^{\mathsf{T}}, \quad \bm{\alpha} \in \mathbb{R}^{N_s-d}.
\end{align}
\end{subequations}
%--------------
Regarding the decomposition in \eqref{eq:alphatotal}, we have Lemma \ref{lem:1}. 
%--------------
\begin{lem}\label{lem:1}
    Let vector $\Tilde{a}$ decomposed into $\Tilde{a}_1,\Tilde{a}_2$ as defined in \eqref{eq:alphatotal}. Then, we have the following 
    \begin{align}
         \bm{a}_n^{\mathsf{T}}\bm{B}_k = \bm{e}_n^{\mathsf{T}} \otimes c_k  + \bm{e}_{n-1}^{\mathsf{T}} \otimes r_k,
    \end{align}
    where $c_k :=  \tilde{\bm{a}}_2^{\mathsf{T}} \tilde{\bm{E}}^{d_k}\bm{g}$ and  $r_k := \tilde{\bm{a}}_1^{\mathsf{T}}(\tilde{\bm{E}}^{N_s-d_k})^\mathsf{T}\bm{g}$ and $\tilde{\bm{E}}\in \mathbb{R}^{N_s\times N_s}$ is defined as
%--------------
\begin{equation}\label{eq:Etilde}
   \tilde{\bm{E}} : = \begin{bmatrix}
            \bm{0}_{N_s-1}^{\mathsf{T}} & 0 \\
            \bm{I}_{N_s-1} &  \bm{0}_{N_s-1}
    \end{bmatrix}.
\end{equation}
%-------------- 
\end{lem}
\begin{proof}
See Appendix \ref{sec:proof}.  
\end{proof}
%--------------
Using Lemma \ref{lem:1}, the equality in \eqref{eq:a_nprod} becomes 
%--------------
\begin{align}
    \label{eq:ckplusrk}
    \bm{e}_n^{\mathsf{T}} \otimes c_k  + \bm{e}_{n-1}^{\mathsf{T}} \otimes r_k =  \bm{e}_n^{\mathsf{T}}, \quad k\in [K].
\end{align}
%--------------
To satisfy \eqref{eq:ckplusrk}, we have to set $c_k$ and $r_k$ as follows
%--------------
\begin{subequations}
\begin{align}
\label{eq:innerProda_E}
        c_k &= 1,  \\
        \label{eq:c_k}
        r_k & = 0, 
\end{align}
\end{subequations}
%--------------
for all $k\in [K]$. The condition in \eqref{eq:c_k} simply results in 
 %--------------
 \begin{align}\label{eq:alpha_1_zero}
  r_k = 0 \Rightarrow \tilde{\bm{a}}_1^{\mathsf{T}}(\tilde{\bm{E}}^{N_s-d_k})^\mathsf{T}\bm{g} = 0 \Rightarrow \tilde{\bm{a}}_1 = \bm{0},
 \end{align}
 %--------------
 accordingly,  we have $\tilde{\bm{a}} = \tilde{\bm{a}}_2$. 
%  Due to the random time-shift $d$, the summation of the first $d$ elements of the inner product in \eqref{eq:innerProda_E} can be corrupted by inter-symbol interference. To overcome this issue, we propose discarding them by appointing zero values to the first $d$ elements of vector $\bm{a}_n$, i.e., 
% %--------------
% \begin{align}\label{eq:alpha_tilde}
%     \tilde{\bm{a}} = [\bm{0}_{d}^{\mathsf{T}}, \bm{\alpha}^{\mathsf{T}}],
% \end{align}
% %--------------
% where $\bm{\alpha}\in \mathbb{R}^{N_s-d}$. 
Then, for the next condition in \eqref{eq:innerProda_E},  we have $ \tilde{\bm{a}}_2^{\mathsf{T}} \tilde{\bm{E}}^{d_k}\bm{g} = 1$ for all $k \in [K]$, which results in the following equation system
%--------------
\begin{equation}
        \sum\nolimits_{n=1}^{N_s-d}{\alpha}[n]g[d_k+n] = 1, \quad d_k  = 0, 1, \ldots, d, \forall k
\end{equation}
%--------------
or equivalently, 
%--------------
\begin{equation}\label{eq:unbias_const}
    \mathscr{H}(\bm{g})\bm{\alpha} = \mathds{1}_{d+1},
\end{equation}
%--------------
where $\mathscr{H}(\cdot)$ is the Hankel lifting operator defined in \eqref{eq:Hankel}.
%--------------
%----- start: Remove convolution stuff until we find a solution
\iffalse
which, in turn, is equivalent to
%--------------
\begin{align}
    \label{eq:Conv}
    \bm{g} \circledast  \bm{\alpha} = \mathds{1}_{N_s-d},
\end{align}
%--------------
where $\bm{g}$ is the reverse sorted version of the pulse shaping filter $\bm{g}$, and $\mathds{1}_{N_s-d}$ is the $1$-vector of size $(N_s-d)\times 1$. 
\fi
%----- end: Remove convolution stuff until we find a solution
Any vector $\bm{\alpha}$ satisfying \eqref{eq:unbias_const} yields an unbiased estimator with \eqref{eq:fMeasure}. Therefore, we formulate Problem $\mathcal{P}_2$ as
%--------------
\begin{align}
    \label{eq:ProblemP2}
     \mathcal{P}_2 ={\rm find} \quad \bm{\alpha},
    \quad   {\rm s.t.} \quad \mathscr{H}(\bm{g})\bm{\alpha} = \mathds{1}_{d+1},
\end{align}
%--------------
where a solution $\bm{\alpha}$ to Problem $\mathcal{P}_2$ also solves $\mathcal{P}_1$ through the following steps: 1) use \eqref{eq:alpha2} to get $\tilde{\bm{\alpha}}_2$, 2) use \eqref{eq:alpha_1_zero} to get $\tilde{\bm{\alpha}}=\bm{0}+\tilde{\bm{\alpha}}_2$, 3) use \eqref{eq:a_ne_n} to get the $\bm{a}_n$'s and finally 4) form $\bm{A}$ with the $\bm{a}_n$'s as rows. We give the condition under which Problem $\mathcal{P}_2$ is feasible in Proposition \ref{prop:existence}.
%----- start: Remove convolution stuff until we find a solution
\iffalse
%--------------
\begin{align}
    \label{eq:HanG}
    \mathscr{H}(\bm{g})\bm{\alpha} = \mathds{1}_{N_s-d},
\end{align}
%--------------
where $\mathscr{H}(\cdot)$ denotes Hankel lifting operator which is defined in \eqref{eq:Hankel}. 
\fi
%----- end: Remove convolution stuff until we find a solution
%-------------------
\begin{prop}\label{prop:existence}
Let the discrete time delays be restricted to $d_k\leq d$, the maximum delay be nonzero $d > 0$, and $\bm{g}$ be the pulse-shaping filter with $N_s$ taps. Then, the linear system of equations in \eqref{eq:unbias_const} has at least one solution if and only if the number of filter taps $N_s$ is greater than $d + {\rm rank}(\mathscr{H}(\bm{g}))$, i.e.,
%--------------
\begin{align}
    \label{eq:upperbound}
    N_s \geq  d + {\rm rank}(\mathscr{H}(\bm{g})),
\end{align}
%--------------
where $\mathscr{H}(\bm{g})$ is the Hankel structure of the filter $\bm{g}$.
\end{prop}
%-------------------
%-------------------
\begin{cor}
 For a rectangular pulse shape, i.e., $\bm{g}=\mathds{1}_{N_s}$, the upper bound in \eqref{eq:upperbound} becomes $N_s \geq  d + 1$ or equivalently, 
 %--------------
 \begin{align}
      d \leq N_s -1,
 \end{align}
 %--------------
 where it comes from the fact that ${\rm rank}(\mathscr{H}(\bm{g})) = 1$ for the rectangular pulse. In other words, for the rectangular pulse, a single sample is enough to get an unbiased estimate.
\end{cor}
%-------------------'
%-------------------
\begin{cor}\label{eq:cor_general}
For a general pulse shape $\bm{g}$, a sufficient condition for delay compensation is that the maximum delay is upper bounded as
 %--------------
 \begin{align}
      d \leq \Big\lfloor \frac{N_s-1}{2}\Big\rfloor.
 \end{align}
 %--------------
Here, we use the fact that ${\rm rank}(\mathscr{H}(\bm{g})) \leq d+1$.

\end{cor}

%-------------------

By solving $\mathcal{P}_2$ we get a filter $\bm{A}$ which yields unbiased function estimation. However, in \eqref{eq:ProblemP2} there is no restriction on the norm of $\bm{a}_n$, which can lead to a large noise-induced error. Therefore, we propose minimization of $\|\bm{a}_n\|$ while retaining a small bias by regularizing over $\|\mathcal{H}(\bm{g})\bm{\alpha} - \mathds{1}_{d+1}\|$. We formulate this as a Tikhonov regularization problem
%-------------------
\begin{equation}
\label{eq:Prob3}
   \mathcal{P}_3 =  \min_{\bm{\alpha}}   \quad \|\mathscr{H}(\bm{g})\bm{\alpha} - \mathds{1}_{d+1}\|_2^2 + \lambda \| \bm{\alpha}\|_2^2,
\end{equation}
%-------------------
 where $\lambda > 0$ represents the regularization parameter which makes a tradeoff between the synchronization-induced bias and the variance caused by the noise. Tikhonov regularization is convex, and the closed-form solution to $\mathcal{P}_3$ is the generalized Moore–Penrose inverse of matrix $\mathscr{H}(\bm{g})$ \cite[Eq. 6.10]{boyd2004convex}, i.e.,
%-------------------
\begin{equation}\label{eq:proposed_filter}
    \bm{\alpha}^* = (\mathscr{H}(\bm{g})^{\mathsf{T}}\mathscr{H}(\bm{g}) + \lambda \bm{I}_{N_s-d} )^{-1}\mathscr{H}(\bm{g})^{\mathsf{T}}\mathds{1}_{d+1}. 
\end{equation}
%-------------------
Note that while computing the inverse matrix in \eqref{eq:proposed_filter} can be managed using the low-rank tensor decomposition method \cite{lee2016regularized}, Problem $\mathcal{P}_3$ needs to be solved only once and offline. % at \ac{FC}.   

In the next section, we assess the performance of \eqref{eq:proposed_filter}. 

% %===================================================
\section{Numerical Results}\label{sec:num}
% %===================================================
For the numerical results, we compare the bias and \ac{MSE} of the solution of our proposed filter in Problem $\mathcal{P}_3$  to the Goldenbaum’s scheme~\cite{goldenbaum2013robust}, i.e., when  $\bm{a} = \bm{g}$. The simulation code is available at \url{https://github.com/henrikhellstrom93/FilterAirComp}. In particular, we randomly generate a sequence of $N=10$ messages for each device from the compact set $\mathbb{D}_f \in [0, 3]$. These messages are pulse shaped with $\bm{g}$ as the Gaussian filter, normalized to $\|\bm{g}\|_2 = 1$. We consider $K=100$ devices, an $M=10$-length random phase sequence, and regularization parameter $\lambda=0.1$. Moreover, the noise of the channel has variance $\sigma_z^2=1$ and the channel coefficients are generated as Rayleigh fading coefficients with $h_k \sim \mathcal{C}\mathcal{N}(0,1)$ for $k\in [K]$. However, note that only the angle $\angle h_k$ is relevant for the simulation since channel magnitude is compensated for perfectly. We further evaluate the bias and \ac{MSE} over $N_{mc} = 10,000$ Monte-Carlo trials as 
\begin{align}
    \text{bias} &= \frac{1}{N_{mc}}\sum\nolimits_{j=1}^{N_{mc}}(\text{mean}(\bm{f}^{(j)}-\hat{\bm{f}}^{(j)})) \\
    \text{MSE} &= \frac{1}{N_{mc}}\sum\nolimits_{j=1}^{N_{mc}}(\text{mean}(\bm{f}^{(j)}-\hat{\bm{f}}^{(j)})^2),
\end{align}
where $\bm{f}^{(j)}$ is the $j$-th Monte-Carlo trial evaluated by applying \eqref{eq:function} to the randomly generated messages and $\hat{\bm{f}}^{(j)}$ denotes the the corresponding estimated value calculated with \eqref{eq:final_est} using the filter $\bm{a}_n$. 
% We will release the simulation code upon publication of the article.

In Fig \ref{fig:bigdelay}, we consider a scenario where the time synchronization error is high compared to the number of samples per symbol. Specifically, $N_s = 2d+2$, one sample more than the minimum requirement in Corollary \ref{eq:cor_general}. In this scenario, our proposed filter offers significantly better performance than the Goldenbaum method, both in terms of bias and \ac{MSE}. 

% ==================
\begin{figure}[!t]
\centering
    \begin{tikzpicture} 
    \begin{axis}[
        xlabel={$d$},
        ylabel={${\rm MSE}/{\rm bias}^{2}$},
        label style={font=\scriptsize},
        legend cell align={left},
        tick label style={font=\scriptsize} , 
        width=8cm,
        height=5.5cm,
        xmin=0, xmax=10,
        ymin=5e-4, ymax=0.77,
        %xtick={-1, -0.5, 0, 0.5, 1},
        %ytick={0, 20, 40, 60, 80, 100},
        % ymode = log,
       legend style={nodes={scale=0.45, transform shape}, at={(0.39,0.97)}}, 
        ymajorgrids=true,
        xmajorgrids=true,
        grid style=dashed,
        grid=both,
        grid style={line width=.1pt, draw=gray!15},
        major grid style={line width=.2pt,draw=gray!40},
    ]
% ==========
    \addplot[
        color=antiquefuchsia,
        mark=o,
        mark options = {rotate = 180},
        line width=1pt,
        mark size=2pt,
        ]
    table[x=d,y=MSE]
    {Data/Sim1.dat};
% ==========
  \addplot[
        color=cssgreen,
        mark=o,
        line width=1pt,
        mark size=2pt,
        ]
    table[x=d,y=MSE_mf]
    {Data/Sim1.dat};
% ==========
     \addplot[
        color=antiquefuchsia,
        mark=square,
        line width=1pt,
        mark size=2pt,
        ]
    table[x=d,y=bias]
    {Data/Sim1.dat};
% ==========
     \addplot[
    color=cssgreen,
    mark=square,
    line width=1pt,
    mark size=2pt,
    ]
table[x=d,y=bias_mf]
{Data/Sim1.dat};
    \legend{${\rm Our~MSE}$,${\rm Goldenbaum~MSE}$,${\rm Our~ bias}^{2}$,${\rm Goldenbaum~bias}^{2}$};
    \end{axis}
\end{tikzpicture}
  \caption{Performance comparison of our proposed filter (purple lines) and the classic matched filter (MF, green lines) in a high delay scenario ($N_s = 2d+2$). When the delay is long in comparison to the length of the pulse, the matched filter introduces a large bias to the function estimate. Our proposed filter nearly eliminates this bias while maintaining a relatively low variance. }
  \label{fig:bigdelay}
\end{figure}
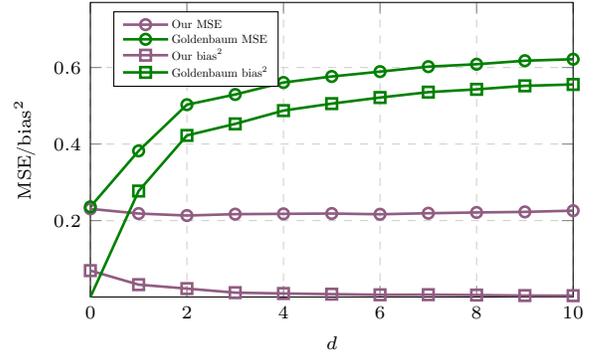

% ==================

In Fig~\ref{fig:smalldelay}, we consider the opposite scenario, i.e., that the delay is short in comparison to the pulse length ($N_s = 2d+20\gg d$). As expected, the performance of the Goldenbaum method is significantly better with long pulses. In this scenario, our proposed filter offers negligible improvement in terms of \ac{MSE}. However, as $d$ increases, the shape of the received waveform loses similarity with the matched filter, and the bias grows. Our proposed filter is better able to match the received waveform, thereby maintaining a low bias.

% ==================
\begin{figure}[!t]
\centering
    \begin{tikzpicture} 
    \begin{axis}[
        xlabel={$d$},
        ylabel={${\rm MSE}/{\rm bias}^{2}$},
        label style={font=\scriptsize},
        legend cell align={left},
        tick label style={font=\scriptsize} , 
        width=8cm,
        height=5.5cm,
        xmin=0, xmax=10,
        ymin=5e-4, ymax=0.38,
        %xtick={-1, -0.5, 0, 0.5, 1},
        %ytick={0, 20, 40, 60, 80, 100},
        % ymode = log,
       legend style={nodes={scale=0.6, transform shape}, at={(0.45,0.97)}}, 
        ymajorgrids=true,
        xmajorgrids=true,
        grid style=dashed,
        grid=both,
        grid style={line width=.1pt, draw=gray!15},
        major grid style={line width=.2pt,draw=gray!40},
    ]
% ==========
    \addplot[
        color=antiquefuchsia,
        mark=o,
        mark options = {rotate = 180},
        line width=1pt,
        mark size=2pt,
        ]
    table[x=d,y=MSE]
    {Data/Sim2.dat};
% ==========
  \addplot[
        color=cssgreen,
        mark=o,
        line width=1pt,
        mark size=2pt,
        ]
    table[x=d,y=MSE_mf]
    {Data/Sim2.dat};
% ==========
     \addplot[
        color=antiquefuchsia,
        mark=square,
        line width=1pt,
        mark size=2pt,
        ]
    table[x=d,y=bias]
    {Data/Sim2.dat};
% ==========
     \addplot[
    color=cssgreen,
    mark=square,
    line width=1pt,
    mark size=2pt,
    ]
table[x=d,y=bias_mf]
{Data/Sim2.dat};
    \legend{${\rm Our~MSE}$,${\rm Goldenbaum~MSE}$,${\rm Our~ bias}^{2}$,${\rm Goldenbaum~bias}^{2}$};
    \end{axis}
\end{tikzpicture}
  \caption{Performance comparison of our proposed filter (purple lines) and the classic matched filter (green lines) in a low delay scenario ($N_s = 2d+20$). When the delay is short compared to the length of the pulse, the matched filter performs quite well, and the improvement of our proposed filter is negligible in terms of \ac{MSE}. However, our proposed filter still achieves a significantly lower bias. }
  \label{fig:smalldelay}
\end{figure}
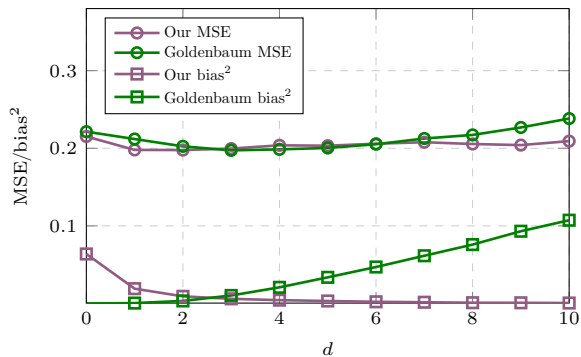

% ==================

In both figures, the proposed filter has a higher bias than the Goldenbaum method when the maximum delay $d$ is close to zero. This is because the variance of the estimator is the dominant source of \ac{MSE} for smaller delays. Therefore the Tikhonov regularizer opts to reduce the norm of the filter taps in favor of reducing the bias, but this can be adjusted via $\lambda$.

%===================================================
\section{Conclusion}\label{sec:conclusion}
%===================================================

In this study, we investigate \ac{OAC} under the assumptions of time-synchronization errors and unknown relative phases between the superimposed signals. Under these assumptions, neighboring symbols cause inter-symbol interference, which is not handled appropriately by the matched filter. Therefore, we propose a new receive filter that achieves unbiased function computation without knowledge of the time delays. Additionally, we formulate a Tikhonov regularization problem that produces an optimal filter given a tradeoff between the bias and noise-induced variance of the function estimates. Numerically, we show that this filter outperforms the matched filter in terms of both bias and \ac{MSE}, especially when the time delays are large in comparison to the length of the transmitted pulses.

\appendix

\subsection{Proof of Lemma \ref{lem:1}}\label{sec:proof}

Using the definition of $\Tilde{a}_1$ and $\Tilde{a}_2$ in \eqref{eq:a_1}, we can rewrite the product of $\bm{a}_n^{\mathsf{T}}\bm{B}_k$ as
%--------------
\begin{align}
    \nonumber
    \bm{a}_n^{\mathsf{T}}\bm{B}_k & = (\bm{e}_n^{\mathsf{T}} \otimes \tilde{\bm{a}}^{\mathsf{T}}) \bm{E}^{d_k}\bm{G}\\
    \nonumber
    & = (\bm{e}_n^{\mathsf{T}} \otimes (\tilde{\bm{a}}_1^{\mathsf{T}}+\tilde{\bm{a}}_2^{\mathsf{T}} )) \bm{E}^{d_k}\bm{G} \\
    & = \Big( (\bm{e}_n^{\mathsf{T}} \otimes \tilde{\bm{a}}_1^{\mathsf{T}})\bm{E}^{d_k} + (\bm{e}_n^{\mathsf{T}} \otimes \tilde{\bm{a}}_2^{\mathsf{T}})\bm{E}^{d_k}\Big) \bm{G}.  \label{eq:a1a2}
\end{align}
%--------------
For the first term, one can check
%--------------
\begin{align}
    (\bm{e}_n^{\mathsf{T}} \otimes \tilde{\bm{a}}_1^{\mathsf{T}}) \bm{E}^{d_k}  = \bm{e}_{n-1}^{\mathsf{T}}\otimes \tilde{\bm{a}}_1^{\mathsf{T}}(\tilde{\bm{E}}^{N_s-d_k})^\mathsf{T}, \label{eq:a_1}
\end{align}
%--------------
where $\tilde{\bm{E}}\in \mathbb{R}^{N_s\times N_s}$ is defined as \eqref{eq:Etilde}.
For the second term, we have
%--------------
\begin{align}
    (\bm{e}_n^{\mathsf{T}} \otimes \tilde{\bm{a}}_2^{\mathsf{T}}) \bm{E}^{d_k}  = \bm{e}_n^{\mathsf{T}} \otimes \tilde{\bm{a}}_2^{\mathsf{T}}\tilde{\bm{E}}^{d_k}.   \label{eq:a_2}
\end{align}
%--------------
Substituting \eqref{eq:a_1} and \eqref{eq:a_2} into \eqref{eq:a1a2}, we obtain
%--------------
\begin{align}
    \nonumber
    \bm{a}_n^{\mathsf{T}}\bm{B}_k & = \bm{e}_n^{\mathsf{T}} \otimes \tilde{\bm{a}}_2^{\mathsf{T}}\tilde{\bm{E}}^{d_k}\bm{G} + \bm{e}_{n-1}^{\mathsf{T}}\otimes \tilde{\bm{a}}_1^{\mathsf{T}}(\tilde{\bm{E}}^{N_s-d_k})^\mathsf{T}(\bm{I}_N\otimes \bm{g}) \\ \nonumber
    & = \bm{e}_n^{\mathsf{T}} \otimes \tilde{\bm{a}}_2^{\mathsf{T}}\tilde{\bm{E}}^{d_k}(\bm{I}_N\otimes \bm{g}) + \bm{e}_{n-1}^{\mathsf{T}}\otimes \tilde{\bm{a}}_1^{\mathsf{T}}(\tilde{\bm{E}}^{N_s-d_k})^\mathsf{T}\bm{g}\\ \nonumber
    & = (\bm{e}_n^{\mathsf{T}} \bm{I}_N)\otimes \tilde{\bm{a}}_2^{\mathsf{T}}\tilde{\bm{E}}^{d_k}\bm{g} + \bm{e}_{n-1}^{\mathsf{T}}\otimes r_k \\
    & = \bm{e}_n^{\mathsf{T}}\otimes c_k +  \bm{e}_{n-1}^{\mathsf{T}}\otimes r_k, 
\end{align}
%--------------
where $c_k : = \tilde{\bm{a}}_2^{\mathsf{T}}\tilde{\bm{E}}^{d_k}\bm{g}$  and $r_k :=\tilde{\bm{a}}_1^{\mathsf{T}}(\tilde{\bm{E}}^{N_s-d_k})^\mathsf{T}\bm{g}$. Hence, we concluded the proof. 

\bibliographystyle{IEEEtran}
\bibliography{IEEEabrv,Ref}

\begin{thebibliography}{10}
\providecommand{\url}[1]{#1}
\csname url@samestyle\endcsname
\providecommand{\newblock}{\relax}
\providecommand{\bibinfo}[2]{#2}
\providecommand{\BIBentrySTDinterwordspacing}{\spaceskip=0pt\relax}
\providecommand{\BIBentryALTinterwordstretchfactor}{4}
\providecommand{\BIBentryALTinterwordspacing}{\spaceskip=\fontdimen2\font plus
\BIBentryALTinterwordstretchfactor\fontdimen3\font minus \fontdimen4\font\relax}
\providecommand{\BIBforeignlanguage}[2]{{%
\expandafter\ifx\csname l@#1\endcsname\relax
\typeout{** WARNING: IEEEtran.bst: No hyphenation pattern has been}%
\typeout{** loaded for the language `#1'. Using the pattern for}%
\typeout{** the default language instead.}%
\else
\language=\csname l@#1\endcsname
\fi
#2}}
\providecommand{\BIBdecl}{\relax}
\BIBdecl

\bibitem{gastpar2003source}
M.~Gastpar and M.~Vetterli, ``Source-channel communication in sensor networks,'' in \emph{Info. Proc. in Sensor Net.}\hskip 1em plus 0.5em minus 0.4em\relax Springer, 2003.

\bibitem{abari2016over}
O.~Abari \emph{et~al.}, ``Over-the-air function computation in sensor networks,'' \emph{arXiv preprint arXiv:1612.02307}, 2016.

\bibitem{csahin2023survey}
A.~{\c{S}}ahin and R.~Yang, ``A survey on over-the-air computation,'' \emph{IEEE Communications Surveys \& Tutorials}, April 2023.

\bibitem{elsts2016microsecond}
A.~Elsts \emph{et~al.}, ``Microsecond-accuracy time synchronization using the {IEEE} 802.15. 4 {TSCH} protocol,'' in \emph{IEEE LCN Workshops}, 2016.

\bibitem{5gspecification}
3GPP, ``{3GPP} technical specification {TS} 38.104, {NR}; base station radio transmission and reception.''

\bibitem{papananos1999radio}
Y.~E. Papananos, \emph{{Radio-Frequency Microelectronic Circuits for Telecommunication Applications}}.\hskip 1em plus 0.5em minus 0.4em\relax Springer Science \& Business Media, 1999.

\bibitem{zhu2019broadband}
G.~Zhu \emph{et~al.}, ``Broadband analog aggregation for low-latency federated edge learning,'' \emph{IEEE Trans. Wireless Commun.}, 2019.

\bibitem{cao2020optimized}
X.~Cao \emph{et~al.}, ``Optimized power control for over-the-air computation in fading channels,'' \emph{IEEE Trans. on~Commun.}, 2020.

\bibitem{liu2020over}
W.~Liu \emph{et~al.}, ``Over-the-air computation systems: Optimization, analysis and scaling laws,'' \emph{IEEE Trans. Wireless Commun.}, 2020.

\bibitem{zang2020over}
X.~Zang \emph{et~al.}, ``Over-the-air computation systems: Optimal design with sum-power constraint,'' \emph{IEEE Wireless Commun. Letters}, 2020.

\bibitem{hellstrom2022unbiased}
H.~Hellstrom \emph{et~al.}, ``Unbiased over-the-air computation via retransmissions,'' in \emph{IEEE Globecom Conference}.\hskip 1em plus 0.5em minus 0.4em\relax IEEE, 2022.

\bibitem{hellstrom2023retransmission}
------, ``Federated learning over-the-air by retransmissions,'' \emph{IEEE Trans. Wireless Commun.}, 2023.

\bibitem{razavikia2023computing}
S.~Razavikia \emph{et~al.}, ``Computing functions over-the-air using digital modulations,'' \emph{arXiv preprint arXiv:2303.00577}, 2023.

\bibitem{shao2021federated}
Y.~Shao \emph{et~al.}, ``Federated edge learning with misaligned over-the-air computation,'' \emph{IEEE Trans. Wireless Commun.}, 2021.

\bibitem{razavikia2022blind}
S.~Razavikia \emph{et~al.}, ``Blind asynchronous over-the-air federated edge learning,'' in \emph{IEEE Globecom Workshops}, 2022.

\bibitem{goldenbaum2009function}
M.~Goldenbaum \emph{et~al.}, ``On function computation via wireless sensor multiple-access channels,'' in \emph{Wire. Commun. and Net. Conference}, 2009.

\bibitem{goldenbaum2013robust}
M.~Goldenbaum and S.~Stanczak, ``Robust analog function computation via wireless multiple-access channels,'' \emph{IEEE Trans. on~Commun.}, 2013.

\bibitem{kortke2014analog}
A.~Kortke \emph{et~al.}, ``Analog computation over the wireless channel: A proof of concept,'' in \emph{IEEE, SENSORS}, 2014.

\bibitem{csahin2022over}
A.~{\c{S}}ahin and R.~Yang, ``Over-the-air computation over balanced numerals,'' in \emph{IEEE Globecom Workshops}, 2022.

\bibitem{csahin2023distributed}
A.~{\c{S}}ahin, ``Distributed learning over a wireless network with non-coherent majority vote computation,'' \emph{IEEE Trans. Wireless Commun.}, 2023.

\bibitem{csahin2022demonstration}
------, ``A demonstration of over-the-air computation for federated edge learning,'' in \emph{IEEE Globecom Workshops}, 2022.

\bibitem{Golden2013Harnessing}
M.~Goldenbaum \emph{et~al.}, ``Harnessing interference for analog function computation in wireless sensor networks,'' \emph{IEEE Trans.~Sig.~Proc.}, 2013.

\bibitem{boyd2004convex}
S.~Boyd \emph{et~al.}, \emph{Convex optimization}.\hskip 1em plus 0.5em minus 0.4em\relax Cambridge university press, 2004.

\bibitem{lee2016regularized}
N.~Lee and A.~Cichocki, ``Regularized computation of approximate pseudoinverse of large matrices using low-rank tensor train decompositions,'' \emph{SIAM Journal on Matrix Analysis and Applications}, vol.~37, no.~2, pp. 598--623, 2016.

\end{thebibliography}

\subsection{Copyright Notice}\label{sec:copyright}
This paper has been accepted for publication in the proceedings of GLOBECOM 2023. The following copyright applies: ©2023 IEEE. Personal use of this material is permitted. Permission
from IEEE must be obtained for all other uses, in any current or future
media, including reprinting/republishing this material for advertising or
promotional purposes, creating new collective works, for resale or
redistribution to servers or lists, or reuse of any copyrighted
component of this work in other works

\end{document}